\documentclass[twocolumn]{autart}  

\usepackage{amssymb,latexsym,amsfonts,amsmath}

\usepackage{wrapfig} % for wrapping text around figures

\newenvironment{authorbio}[2][]{%
	\par\addvspace{1em}%
	\noindent
	\if\relax\detokenize{#1}\relax % check if photo is empty
	\else
	\begin{wrapfigure}{l}{0.8in} % narrower than image width
		\vspace{-5pt} % tweak vertical alignment if needed
		\includegraphics[width=1in,height=1.25in,clip,keepaspectratio]{#1}
	\end{wrapfigure}%
	\fi
	\noindent\textbf{#2} % NAME
}{\par\addvspace{2em}}
\usepackage{mathrsfs}

\usepackage{natbib}

\usepackage{xspace}

\usepackage{algorithm}
\usepackage{algorithmic}
\usepackage{graphicx}

\usepackage{multirow}
\usepackage{booktabs}
\usepackage[dvipsnames]{xcolor}

\usepackage{url}  

\makeatletter
\edef\endfrontmatter{%
	\unexpanded\expandafter{\endfrontmatter}% original code
	\noexpand\endNoHyper                  % add this to re-enable links
}
\makeatother

\DeclareRobustCommand\sampleline[1]{%
	\tikz\draw[#1] (0,0) (0,\the\dimexpr\fontdimen22\textfont2\relax)
	-- (2em,\the\dimexpr\fontdimen22\textfont2\relax);%
}

\usepackage{caption}

\usepackage{paralist}
\usepackage{dsfont}
\usepackage{caption} 
\usepackage{booktabs}
\usepackage{eso-pic}
\usepackage{epsfig}
\usepackage{mathtools}
\usepackage {url}
\usepackage{epstopdf}
\usepackage{tikz}
\usepackage{IEEEtrantools}
\usetikzlibrary{arrows,calc,decorations.markings,math,arrows.meta}
\newtheorem{theorem}{Theorem}[section]
\newtheorem{lemma}[theorem]{Lemma}
\newtheorem{problem}[theorem]{Problem}
\newtheorem{proposition}[theorem]{Proposition}

\newtheorem{definition}[theorem]{Definition}
\newtheorem{example}[theorem]{Example}

\newtheorem{remark}[theorem]{Remark}

\numberwithin{equation}{section}

\usepackage{dsfont}
\usepackage{caption}

\usepackage[many]{tcolorbox}
\usetikzlibrary{calc}
\tcbuselibrary{skins}

\newtcolorbox{resp}[1][]{%
	enhanced jigsaw,%
	colback=gray!5!white,%
	colframe=gray!80!black,%
	size=small,%
	boxrule=1pt,%
	halign title=flush center,%
	coltitle=black,%
	breakable,%
	drop shadow=black!50!white,%
	attach boxed title to top left={xshift=1cm,yshift=-\tcboxedtitleheight/2,yshifttext=-\tcboxedtitleheight/2},%
	minipage boxed title=3cm,%
	boxed title style={%
		colback=white,%
		size=fbox,%
		boxrule=1pt,%
		boxsep=2pt,%
		underlay={%
			\coordinate (dotA) at ($(interior.west) + (-0.5pt,0)$);
			\coordinate (dotB) at ($(interior.east) + (0.5pt,0)$);
			\begin{scope}[gray!80!black]
				\fill (dotA) circle (2pt);
				\fill (dotB) circle (2pt);
			\end{scope}
		}%
	},%
	#1%
}

\usepackage{subcaption}
\newcommand{\EE}{\mathds{E}}
\newcommand{\PP}{\mathds{P}}

\usepackage[bb=boondox]{mathalfa}

\usepackage{enumitem}% http://ctan.org/pkg/enumitem

\makeatletter
\long\def\@maketablecaption#1#2{\@tablecaptionsize
    \global \@minipagefalse
    \hbox to \hsize{\parbox[t]{\hsize}{\centering #1 \\ #2}}}

\makeatother

\begin{document}

\begin{frontmatter}
\title{Safety Controller Synthesis for Stochastic Polynomial Time-Delayed  Systems}
\author{Omid Akbarzadeh}\ead{omid.akbarzadeh@newcastle.ac.uk},
\author{MohammadHossein Ashoori}\ead{m.ashoori2@newcastle.ac.uk},
\author{Amy Nejati}\ead{amy.nejati@newcastle.ac.uk}, and  
\author{Abolfazl Lavaei}\ead{abolfazl.lavaei@newcastle.ac.uk}  
\address{School of Computing, Newcastle University, United Kingdom}                           


\begin{abstract}  
This work develops a theoretical framework for safety controller synthesis in discrete-time \emph{stochastic} nonlinear polynomial systems subject to \emph{time-invariant delays} (dt-SNPS-td). While safety analysis of stochastic systems using control barrier certificates (CBC) has been widely studied, developing safety controllers for stochastic systems with time delays remains largely unexplored. The main challenge arises from the need to account for the influence of delayed components when formulating and enforcing safety conditions. To address this, we employ \emph{Krasovskii} control barrier certificates, which extend the conventional CBC framework by augmenting it with an additional summation term that captures the influence of delayed states. This formulation integrates both the current and delayed components into a unified barrier structure, enabling safety synthesis for stochastic systems with time delays. The proposed approach synthesizes safety controllers under \emph{input constraints}, offering probabilistic safety guarantees robust to such delays: it ensures that all trajectories of the dt-SNPS-td remain within the prescribed safe region while fulfilling a quantified probabilistic bound. To achieve this, our method reformulates the safety constraints as a sum-of-squares optimization program, enabling the systematic construction of \emph{Krasovskii} CBC together with their associated safety controllers. We validate the proposed framework through three case studies, \emph{including two physical systems}, demonstrating its effectiveness and practical applicability.
\end{abstract}

\end{frontmatter}
\section{Introduction}\label{sec:Intro}
Time delays are inherent in modern engineering systems, arising from sensing and actuation latencies, computational and scheduling overheads, and communication over shared networks. Such delays are ubiquitous in domains including communications~\citep{Seborg2004} and process control~\citep{Srikant2004}, where even minor timing inconsistencies can critically impact the safety and reliability of system operation. As control and decision-making processes become increasingly distributed and network-dependent, the occurrence of time delays is expected to rise. While often unavoidable, such delays can profoundly alter system behavior: narrowing stability margins, degrading performance, and potentially violating safety constraints that would otherwise hold in delay-free settings.  Within this broad domain, research on safety analysis of discrete-time stochastic nonlinear systems with delays remains limited due to the inherent technical challenges involved. Specifically, the synthesis of controllers for such systems is particularly challenging, not only because of the presence of time delays, but also due to the stochastic dynamics that encapsulate uncertainties within the system model.

To tackle safety synthesis for stochastic control systems \emph{without} time delays, existing literature has relied on employing finite abstractions to approximate the original models with simpler representations featuring discrete state sets (see \emph{e.g.,}~\cite{ref1,julius2009approximations,zamani2014symbolic,lavaei2020compositional,lavaei2022automated}). While promising, these abstraction-based methods heavily rely on discretizing the state and input spaces, which often becomes impractical for real-world systems due to the resulting state-explosion problem. In light of this limitation, recent years have witnessed the development of compositional techniques that construct finite abstractions of high-dimensional stochastic systems by systematically combining abstractions of lower-dimensional subsystems (see \emph{e.g.,}~\cite{hahn2013compositional,ref3,nejati2020compositional,lavaei2022scalable,lavaei2022dissipativity}).

As an alternative for enforcing safe controller synthesis in stochastic systems \emph{without} time delays and avoiding state-space discretization, a prominent line of research leverages \emph{control barrier certificates} (CBC), initially introduced {by~\cite{prajna2004safety}}. Analogous to Lyapunov functions, CBC impose conditions on both the function and its evolution along system trajectories. By selecting an initial level set of CBC corresponding to a prescribed set of initial states, one obtains a separation between the unsafe region and all admissible evolutions, thereby providing (probabilistic) safety guarantees~\citep{santoyo2021barrier,nejati2024context}. This framework has been extensively applied to formal verification and controller synthesis in both deterministic and stochastic settings, although delays have not been considered in these studies~\citep{borrmann2015control,ames2019control,zaker2024compositional,jahanshahi2022compositional,lavaei2024scalable,nejati2024reactive}.

\textbf{Related literature on time-delayed systems.} There exists some limited work on safety analysis of dynamical systems with delays; however, these studies are primarily developed for \emph{deterministic} continuous-time models. In this regard, the study {by~\cite{1582846} extends} barrier certificates to time-delay systems using Razumikhin and Krasovskii functions for deterministic nonlinear polynomial dynamics, focusing mainly on safety verification rather than controller synthesis. Another study {by~\cite{Ames_Safety} introduces} safety functionals to certify safety for deterministic time-delay systems without relying on delay approximations. Furthermore, the work {by~\cite{Ames_Safety_1} develops} control barrier functions for deterministic nonlinear systems with state delays, enabling the design of safety filters.
Another study {by~\cite{liu2023safety} extends} barrier-functional methods to time-delay systems with disturbances by introducing the concept of input-to-state safety. In contrast to the aforementioned studies, the crucial influence of \emph{stochastic disturbances} on system safety has yet to be addressed.  A safety controller synthesis method for stochastic networked systems under communication constraints, including delays and packet drops, is recently proposed by~\cite{akbarzadeh2025safety}. It is worth mentioning that there also exists a body of literature addressing the stability of control systems with time delays, which lies beyond the scope of this work (see \emph{e.g.,}~\cite{Safe-Stabilization,1583086,REN2022110563,Fridman}). A data-driven design of Krasovskii control barrier certificates for deterministic uncertain systems is recently studied by~\cite{akbarzadeh2026data}.

\textbf{Key contributions.} Motivated by the discussed challenges, we introduce Krasovskii \emph{quadratic} control barrier certificates (K-QCBC) and Krasovskii \emph{polynomial} control barrier certificates (K-PCBC) for discrete-time \emph{stochastic} nonlinear polynomial systems with time-invariant \emph{delays} (dt-SNPS-td), capturing the joint influence of the \emph{current state} and the \emph{$h$-step history} on system safety. Our proposed framework addresses the problem of designing safety controllers under \emph{input constraints}, which is {non-convex} due to enforcing the underlying constraint on the input, and resolves the resulting bilinearity within a tractable synthesis formulation.  We propose dedicated conditions formulated as a sum-of-squares (SOS) optimization problem, through which one can jointly compute the safety controller and the corresponding K-QCBC or K-PCBC. This formulation provides probabilistic safety guarantees that are robust to system delays.

{\textbf{Organization.} The remainder of the paper is organized as follows. Section \ref{sec:Problem} formally defines the dt-SNPS-td model and outlines the safety specifications adopted throughout the paper. Section \ref{sec:Krasovskii} develops the K-QCBC and derives a finite-horizon probabilistic safety guarantee. It further reformulates the safety conditions as a tractable SOS optimization problem and presents the corresponding implementable synthesis algorithm. Section~\ref{subsec:problem_G} introduces the K-PCBC framework, casts the resulting safety requirements into an SOS optimization program, and describes the associated algorithmic procedure for practical implementation. Section \ref{sec:Case} validates the proposed framework through three case studies, while Section \ref{sec:Conclusion} concludes the paper with a summary of findings and potential directions for future research.}

\textbf{Notation and Preliminaries.}
We denote the sets of real, positive real, and non-negative real numbers, respectively, by $\mathbb{R}, \mathbb{R}^{+}$, and $\mathbb{R}_0^{+}$. Sets of nonnegative and positive integers are denoted by $\mathbb{N}$ and $\mathbb{N}^{+}$. Given $N$ vectors $x_i \in \mathbb{R}^{n_i}, x=\left[x_1 ; \ldots ; x_N\right]$ is a concatenated \emph{column} vector.
Trace of a matrix $A \in \mathbb{R}^{N \times N}$ with diagonal elements $a_1, \ldots, a_N$ is written as $\mathsf{Tr}(A)=$ $\sum_{i=1}^N a_i$. A diagonal matrix in $\mathbb{R}^{n \times n}$ with diagonal entries $a_1, \ldots, a_n$ is signified by $\operatorname{diag}\left(a_1, \ldots, a_n\right)$.  {An identity matrix of dimension $n\times n$, and a zero matrix or vector with compatible size are denoted, respectively, by $\mathbf{I}_n$ and $\mathbf{0}$}. A normal distribution with the mean $\mu$ and the {covariance matrix $\Sigma$ is denoted by $\mathcal{N}(\mu, \Sigma)$}. For a \emph{symmetric} matrix ${P}, {P} \succ 0$ (${P} \succeq 0$) denotes ${P}$ is positive definite (positive semi-definite). A star $(\star)$ in a symmetric matrix presents the transposed element in the symmetric position. The transpose of a matrix $A$ is written as $A^{\top}$.  The empty set is denoted by $\emptyset$. The Cartesian product of two sets $Q$ and $V$ is denoted by $Q \times V$. For a positive integer $n \in \mathbb{N}^+$, the notation $Q^{n+1}$ represents the $(n+1)$-fold Cartesian product of $Q$ with itself. For the set $S \subseteq Q$, the expression $Q \backslash S$ denotes the set of elements in $Q$ not in $S$. Given a system $\Sigma$ and a {property $\Upsilon, \Sigma \models_{\mathcal{T}} \!\! \Upsilon$ denotes that $\Sigma$ satisfies $\Upsilon$ over the time horizon $\mathcal{T}$.}

By ($\Omega, \mathbb{F}_{\Omega}, \PP_{\Omega}$), we denote a probability space of events, where $\Omega$ is the sample space, $\mathbb{F}_{\Omega}$ is a sigma-algebra on $\Omega$, and $\PP_{\Omega}$ is a probability measure. We assume random variables $\mathbb{X}$ are measurable functions, \emph{i.e.,} $\mathbb{X}\!\!:\left(\Omega, \mathbb{F}_{\Omega}\right) \rightarrow\left(S_{\mathbb{X}}, \mathbb{F}_{\mathbb{X}}\right)$, such that $\mathbb{X}$ induces a probability measure on $\left(S_{\mathbb{X}}, \mathbb{F}_{\mathbb{X}}\right)$ as $\operatorname{Prob}\{\mathds{A}\}=\PP_{\Omega}\left\{\mathbb{X}^{-1}(\mathds{A})\right\}$ for any $\mathds{A} \in \mathbb{F}_{\mathbb{X}}$. A topological space $\mathbf{S}$ is said to be Borel, denoted by $\mathbb{B}(\mathbf{S})$, if it is homeomorphic to a Borel subset of a Polish space, \emph{i.e.,} a separable and metrizable space.

\section{Problem Formulation}\label{sec:Problem}
\subsection{Stochastic Nonlinear Polynomial Time-delayed Systems}
We begin by formally introducing discrete-time stochastic nonlinear polynomial systems with time-invariant delays as the primary model for which we aim to derive a safety certificate and the corresponding controller.

\begin{definition}[\textbf{dt-SNPS-td}]\label{def:dt-SNPS-td}
	A discrete-time stochastic nonlinear polynomial system with time-invariant delays (dt-SNPS-td) is defined by
	\begin{align}\notag
		\Sigma\!:\! x_{k+1} &= A(x_k, x_{k-h}) x_k \!+\! A_1(x_k, x_{k-h}) x_{k-h}\\\label{eq:dt-SNPS-td} &~~~ +  G (x_k, x_{k-h}) u_k \!+\! Ew_k,
	\end{align}
	where
	\begin{itemize}[leftmargin=*]
		\item  $h \in \mathbb N^{+}$ is time-invariant delay;
		\item $x_k, x_{k-h} \in X$, with $k \in \mathbb N$, denote the current and delayed states of the system, respectively, while $X \subseteq \mathbb{R}^n$ is a Borel space representing the state set;
		\item $\mathbf{x}_0 = ({x}_{0}, \ldots,{x}_{-h}) \in {X}^{h+1}$ is the given initial state history; 
		\item $u_k \in  U$  is the control input, with $ U$ being a Borel space representing the admissible control set subject to \emph{input constraints}, defined as
		\begin{equation}\label{input_set}
			U=\left\{u \in \mathbb{R}^m  \,\,\,\big|\,\,\, b_j^{\top} u \leq 1, j=1, \ldots, J\right\}\!, \quad b_j \in \mathbb{R}^m,
		\end{equation}
		\item $w_k \in  \mathbb{R}^n$  denotes a sequence of independent and identically distributed (i.i.d.) random variables from a {standard} normal distribution $\mathcal{N}\left({\mathbf{0}, \mathbf{I}_n}\right)$, representing the system's process noise;
		\item $A(x_k, x_{k-h}) \in \mathbb R^{n\times n}$, $A_1(x_k, x_{k-h}) \in \mathbb R^{n\times n}$, and  $G (x_k, x_{k-h}) \in \mathbb R^{n\times m}$ are all \emph{polynomial} matrices, while $E \in  \mathbb R^{n\times n}$ is a constant matrix.
	\end{itemize}
\end{definition}
We consider a state-feedback controller of the form
\begin{equation}\label{feedback}
	u_k= F(x_k, x_{k-h}) x_k + F_1(x_k, x_{k-h}) x_{k-h},
\end{equation}
where $F(x_k, x_{k-h}) \in \mathbb R^{ m \times n}$ and $ F_1(x_k, x_{k-h}) \in \mathbb R^{m \times n}$ are \emph{polynomial} matrices to be designed as part of this work. For simplicity of presentation, we define $\mathcal A = A(x_k, x_{k-h})$, $\mathcal A_1=A_1(x_k, \allowbreak x_{k-h})$, $\mathcal{G}=G (x_k, x_{k-h})$, $\mathcal{F} =F(x_k, x_{k-h})$ and $\mathcal{F}_1 =F_1(x_k, x_{k-h})$ throughout the rest of the paper. We represent dt-SNPS-td in~\eqref{eq:dt-SNPS-td} using the tuple $\Sigma = \left({\mathcal A, \mathcal A_1,  \mathcal{G},\allowbreak E, X, \allowbreak U, h}\right)$. We denote by {$\mathbf{x}_{\mathbf{x}_0uw}(k)$} the 
solution process of dt-SNPS-td $\Sigma$ in Definition~\ref{def:dt-SNPS-td} at time $k \in \mathbb{N}$, together with its $h$-step history, under input and process noise trajectories $u(\cdot)$ and $w(\cdot)$, starting from an initial state history of $\mathbf{x}_0 \in X^{h+1}$.

\begin{remark}
	It is worth noting that the class of systems in Definition~\ref{def:dt-SNPS-td} is restricted to nonlinear polynomial dynamics, as the safety conditions (cf.~\eqref{eq:CBC}) are eventually reformulated into an SOS program (cf. Lemmas~\ref{SOS_K-QCBC}, \ref{sos}) and solved using \textsf{SOSTOOLS}~\citep{prajna2004sostools}. The proposed results, however, can be extended to more general nonlinear systems of the form \(x_{k+1} = f(x_k, x_{k-h}, \allowbreak u_k) + w_k\) which could be solved using nonlinear optimization tools such as \textsf{IPOPT}~\citep{waechter2006ipopt} and \textsf{SNOPT}~\citep{gill2005snopt}. The same reasoning applies to the stochastic component, where assuming a normal noise distribution simplifies the computations in our conditions. The results, however, can be extended to any arbitrary noise distributions, in which case the corresponding expected value (cf.~\eqref{subeq:decreasing}) should be evaluated with respect to the chosen distribution when solving the problem.
\end{remark}

The following illustrative example demonstrates that the dt-SNPS-td $\Sigma$, introduced in Definition~\ref{def:dt-SNPS-td}, can indeed represent a broad class of discrete-time stochastic nonlinear polynomial systems with time-invariant delays, where coupling may occur between \emph{delayed and non-delayed} states and inputs.

\begin{example}\label{example}
	Consider the following two-dimensional system with a time-invariant delay \(h\):
	\begin{align*}
		{x}_{(k+1)_1} &= x_{k_1}x_{k_2} + 2x_{(k-h)_1}x^2_{k_2} - x_{(k-h)_1} + x_{k_1} u_k\\ &~~~ + 0.4w_{k_1},\\
		{x}_{(k+1)_2} &= -x_{k_1} + x^2_{(k-h)_2}+ x_{(k-h)_2} u_k +0.5w_{k_2},
	\end{align*}
	where \(x_k=[x_{k_1},x_{k_2}] \in X \subseteq \mathbb R^2\), \(u_k\in  U\subset\mathbb R\), and {$w_{k}=[w_{k_1},w_{k_2}]\sim\mathcal N(\mathbf{0},\mathbf{I}_2)$}.
	The dynamics take the dt-SNPS-td form of Definition~\ref{def:dt-SNPS-td} as
	\begin{align*}
		{x}_{k+1} &=
		\overbrace{\begin{bmatrix}
				x_{k_2} & x_{(k-h)_1} x_{k_2}\\
				-1 & 0
		\end{bmatrix}}^{A(x_k, x_{k-h})}
		\overbrace{\begin{bmatrix}
				x_{k_1}\\
				x_{k_2}
		\end{bmatrix}}^{x_k} +
		\overbrace{\begin{bmatrix}
				-1 + x^2_{k_2} & 0\\
				0 & 	x_{(k-h)_2}
		\end{bmatrix}}^{A_1(x_k, x_{k-h})}\\ 
	 &\quad\times
		\underbrace{\begin{bmatrix}
				x_{(k-h)_1}\\
				x_{(k-h)_2}
		\end{bmatrix}}_{x_{k-h}}
		+
		\underbrace{\begin{bmatrix}
				x_{k_1}\\
				x_{(k-h)_2}
		\end{bmatrix}}_{G (x_k, x_{k-h})} u_k +
		\underbrace{\begin{bmatrix}
				0.4 & 0\\
				0 & 	0.5
		\end{bmatrix}}_{E}\underbrace{\begin{bmatrix}
				w_{k_1}\\
				w_{k_2}
		\end{bmatrix}}_{w_k}\!\!.
	\end{align*}
	\hfill$\blacksquare$
\end{example}

We now formalize the safety property for the dt-SNPS-td $\Sigma$ introduced in Definition~\ref{def:dt-SNPS-td} through the following definition.

\begin{definition}[\textbf{Safety Property}]\label{safety}
	Given a dt-SNPS-td  $\Sigma={(\mathcal A, \mathcal A_1,  \mathcal{G},\allowbreak E, X, \allowbreak U, h)}$, consider a safety specification $\Upsilon=(X_a, X_b, \mathcal{T})$, where $ X_a, X_b \allowbreak  \subset X$ are its initial and unsafe sets, respectively.
	We assume $X_a \cap  X_b=\emptyset$, otherwise the system is unsafe with a probability of 1. The dt-SNPS-td is said to be safe within time horizon $\mathcal{T} \in \mathbb{N}$, denoted by $\Sigma \models_\mathcal T \Upsilon$, if all trajectories of $\Sigma$ starting from any initial state history $\mathbf{x}_0 \in X_a^{h+1}$ do not reach {$ X_b$ during $\mathcal{T}$ time steps. Equivalently, the state history $(x_k,x_{k-1},\ldots,x_{k-h})$ should not lie in $X_b \times (X\backslash X_b)^h$ for any $k \in \{1,\dots,\mathcal{T}\}$.} Since trajectories of $\Sigma$ are probabilistic, the aim is to formally compute $\PP \{\Sigma \models_\mathcal T \Upsilon\} \geq 1-\mu_{h}$, where $\mu_{h} \in(0,1]$ denotes the safety violation bound.
\end{definition}

For completeness and to provide sufficient intuition to the general reader, we first briefly describe in the next subsection the CBC formulation for stochastic systems without time delays, \emph{i.e.,}

\begin{equation}\label{baseline}
	x_{k+1} \!=\! A(x_k) x_k \!+\!  G (x_k) u_k \!+\! E w_k.
\end{equation}
We then extend this formulation by generalizing the conventional barrier concept to the \emph{Krasovskii} framework in the subsequent sections.

\subsection{Control Barrier Certificates}\label{sec:CBC}
We first present the notion of control barrier certificates, adapted from~\citep{1428804}.

\begin{definition}[{\textbf{CBC}}]\label{def:CBC} 
	Consider the stochastic system in~\eqref{baseline}, with ${X}_a \subset X$ and $ {X}_b \subset X$ being its initial and unsafe sets, respectively. Assuming the existence of constants $\eta, \gamma_a, \gamma_b \in \mathbb{R}^{+}$, with $\gamma_b > \gamma_a$, a function $\mathcal B:  X \to \mathbb{R}_0^+$ is called a control barrier certificate (CBC) for the system in~\eqref{baseline} if
	\begin{subequations}\label{eq:CBC_Base}
		\begin{align}\label{CBC_Base_1}
			&  \:\:  \mathcal B({x}) \leq \gamma_a, \hspace{3cm}  \forall {x} \in {X}_{a},\\\label{CBC_Base_2}
			&  \:\:  \mathcal B({x}) \geq \gamma_b, \hspace{3cm} \forall {x} \in {X}_b,
		\end{align}  
		and $\forall {x} \in {X}, \exists\, u \in {U}$, such that
		\begin{align}\label{CBC_Base_3}
			& \EE\big[ \mathcal B(x_{k+1}) \,\,\,\big|\,\,\, x_k,u_k\big] - \mathcal B({x}_k) \leq  \eta.
		\end{align}
	\end{subequations}
\end{definition}

As the system in~\eqref{baseline} evolves stochastically, building on~\citep{Kushner,1428804} and using Definition~\ref{def:CBC}, the following theorem establishes a lower bound on the probability that the system trajectories in~\eqref{baseline} remain within the safe set over a finite horizon.

\begin{theorem}[\textbf{Probabilistic Safety}]\label{Th1}
	Given the stochastic system in~\eqref{baseline}, assume that there exist a CBC $\mathcal{B}$ and a control policy $u(\cdot)$ satisfying conditions~\eqref{CBC_Base_1}-\eqref{CBC_Base_3}. Then, the probability that the solution process {${x}_{{x}_0uw}(k)$,} starting from any initial state ${x}_0 \in {X}_a$ under $u(\cdot)$  and a process noise $w(\cdot)$, never reaches the unsafe set ${X}_b$ within the time horizon $k \in \{1,\dots,\mathcal{T}\}$ is bounded by
\end{theorem}
\begin{align*}
	&\PP\Big\{{{x}_{{x}_0uw}(k)} \notin X_b \text { for all} ~ {k \in \{1,\dots,\mathcal{T}\}}] \,\,\big|\, {{x}_0} \Big\} \\&~~~~\! \geq \! 1 \! - \! \frac{\gamma_a + \eta \mathcal{T}}{\gamma_b}.
\end{align*}

In this work, to handle the inherent delays in the system, we introduce two types of \emph{Krasovskii} control barrier certificates: the Krasovskii \emph{quadratic} control barrier certificate (K-QCBC) and the Krasovskii \emph{polynomial} control barrier certificate (K-PCBC). Between the two, K-PCBC offers greater flexibility for constructing safety certificates and controllers, as its polynomial structure is not limited to quadratic forms like K-QCBC. However, this generality comes at the cost of increased computational complexity (cf. Table~\ref{tab:system-configurations} in the case study section). In contrast, K-QCBC, as a structured subclass, facilitates certification and controller synthesis through its tractable algebraic form, enabling safety conditions to be expressed as matrix inequalities (cf.~\eqref{K-PCBC}). 

As the first development, we formulate Problem~\ref{Prob1}, which focuses on the safety synthesis of the dt-SNPS-td $\Sigma$ using the K-QCBC framework.

\begin{resp}
\begin{problem}\label{Prob1}
	Given a dt-SNPS-td $\Sigma = {(\mathcal A, \mathcal A_1,  \mathcal{G},\allowbreak E, X, \allowbreak U, h)}$ subject to input constraints~\eqref{input_set}, and a safety specification $\Upsilon=( {X}_a, {X}_b, \mathcal{T})$, synthesize a formal controller based on the K-QCBC framework that guarantees satisfaction of the safety specification $\Upsilon$ over the time horizon $\mathcal{T}$ with a probabilistic bound $\mu_{h} \in(0,1]$, i.e., $\PP \{\Sigma \models_{\mathcal{T}} \Upsilon\} \geq 1-\mu_{h}$, such that the safety guarantee remains robust to the presence of system delays.
\end{problem}
\end{resp}

\section{Krasovskii Quadratic Control Barrier Certificates}\label{sec:Krasovskii}
We denote the sequence of state history at time $k \in \mathbb{N}$ by 
\begin{equation}\label{history}
	\mathbf{x}_k =(x_k,x_{k-1},\ldots,x_{k-h}) \in X^{h+1},   
\end{equation}
and the successor sequence is presented by
\begin{equation}\label{successor}
	\mathbf{x}_{k+1}=(x_{k+1},x_k,\ldots,x_{k-h+1})\in X^{h+1}.
\end{equation}
From \eqref{history}, we define
\begin{equation}\label{history-simple}
	\underbrace{\mathbf{x}_k}_\mathbf{x}=(\underbrace{x_k}_{x},\underbrace{x_{k-1}}_{x_1},\ldots,\underbrace{x_{k-h}}_{x_h}),   
\end{equation}
where $x$ denotes the current state, $x_1$ the one-step delayed state, and $x_h$ the $h$-step delayed state.

To derive sufficient safety conditions for a dt-SNPS-td in~\eqref{eq:dt-SNPS-td}, we consider a Krasovskii quadratic control barrier certificate (K-QCBC) in the form of
\begin{equation}\label{k-CBC}
	\mathcal{B}(\mathbf{x})=x^\top P x + \sum_{i=1}^{h} x^\top_{i} P_1 x_{i}.
\end{equation}
As evident from~\eqref{k-CBC}, the K-QCBC consists of two components: the first corresponds to the conventional quadratic CBC that accounts only for the current state, while the second is a summation over delayed states that captures their influence in the safety synthesis process. This type of barrier certificate in~\eqref{k-CBC} is inspired by the \emph{Lyapunov–Krasovskii} functional approach used in the stability analysis of time-delayed systems~\citep{Papachristodoulou_time_delay}.
{\begin{remark}
While the K-QCBC in \eqref{k-CBC} consists of two independent components, $x_k$ and $x_{k-h}$, it is designed to capture their coupling through the system dynamics. This is evident from \eqref{eq:dt-SNPS-td}, where the polynomial matrices $A(x_k,x_{k-h})$ and $A_1(x_k,x_{k-h})$ depend jointly on both $x_k$ and $x_{k-h}$. We refer to Example~\ref{example} and the academic system dynamics in Appendix~\ref{A1}, where cross-term interactions between $x_k$ and $x_{k-h}$ explicitly appear.
\end{remark}}
{We formally introduce the notion of the K-QCBC in the following definition.}

\begin{definition}[\textbf{K-QCBC}]\label{def:KCBC}
	Consider a dt-SNPS-td
	$\Sigma = {(\mathcal A, \mathcal A_1,  \mathcal{G},\allowbreak E, X, \allowbreak U, h)}$, with $ {X}_a \subset  X$ and $ {X}_b \subset X$ being its initial and unsafe sets, respectively. Assuming the existence of constants $\eta, \gamma_a, \gamma_b \in \mathbb{R}^{+}$, with $\gamma_b > \gamma_a$, a function $\mathcal B: X^{h+1} \to \mathbb{R}_0^+$ is called a Krasovskii quadratic control barrier certificate (K-QCBC) for $\Sigma$ if
	\begin{subequations}\label{eq:CBC}
		\begin{align}
			&  \:\:  \mathcal B(\mathbf{x}) \leq \gamma_a, \hspace{1.5cm}  \forall \mathbf{x} \in {X}^{h+1}_{a},\label{subeq:initial}\\
			&  \:\:  \mathcal B(\mathbf{x}) \geq \gamma_b, \hspace{1.5cm} \forall \mathbf{x} \in X_b \times (X\backslash X_b)^h, \label{subeq:unsafe}
		\end{align}  
		and $\forall \mathbf{x} \in {X}^{h+1}, \exists\, u \in {U}$, such that
		\begin{align}\label{subeq:decreasing}
			& \EE\big[ \mathcal B(\mathbf{x}_{k+1}) \,\,\,\big|\,\,\,  \mathbf{x}_k,u_k\big] - \mathcal B(\mathbf{x}_k) \leq  \eta.
		\end{align}
	\end{subequations}
\end{definition}
\begin{remark}
	Note that preventing the state history from entering the set $X_b^{h+1}$ is insufficient to guarantee safety, as it only prevents trajectories in which \emph{all} $h+1$ states lie in the unsafe region, without ruling out cases where some states may still be unsafe. In contrast, preventing the state history from entering the set $X_b \times (X \setminus X_b)^h$ in~\eqref{subeq:unsafe} ensures that if none of the previous $h$ states lie in the unsafe set, then the current state will also remain safe. Given that the initial state history lies in $X_a^{h+1}$ and $X_a \cap X_b = \emptyset$, this condition guarantees that the system will never reach the unsafe set, thereby ensuring safety.
\end{remark}
Inspired by Theorem~\ref{Th1}, the following proposition builds upon Definition~\ref{def:KCBC} and establishes a lower bound on the probability that the trajectories of the dt-SNPS-td $\Sigma$ remain within the safe set over a finite time horizon.

\begin{proposition}\label{Th:safety}
	Given a dt-SNPS-td $\Sigma=(\mathcal A, \mathcal A_1,  \mathcal{G}, E, \\ X, U, h)$, assume there exists a K-QCBC $\mathcal{B}$ { and a control policy $u(\cdot)$} as in Definition~\ref{def:KCBC}. Then the probability that the solution process {${\mathbf x}_{\mathbf x_0uw}$} of $\Sigma$, starting from any initial history $\mathbf x_0  \in {X}^{h+1}_a$ {under $u(\cdot)$} and a process noise $w(\cdot)$ never reaches $ X_b \times (X\backslash X_b)^h$ within { the time horizon $k \in \{1,\cdots,\mathcal{T}\}$} is quantified as
\end{proposition}
\begin{align*}
	&\PP\Big\{ {\mathbf x}_{{\mathbf  x_0uw}}(k) \!\notin\!  X_b \!\times\! (X\backslash X_b)^h \text { for all } {k\! \in\! \{1,\dots,\!\mathcal{T}\}} \,\big|\, \mathbf{x}_0 \Big\}\\ & ~~~\geq 1- \mu_{h},
\end{align*}
with $\mu_{h} = \frac{\gamma_a + \eta \mathcal{T}}{\gamma_b}$, which is robust to the system's delay $h$.

{\bf Proof.}
	According to the condition~\eqref{subeq:unsafe}, we have
	\begin{align*}
		&X_b \times (X\backslash X_b)^h \subseteq\Big\{ \mathbf{x} \in {X}^{h+1}  \,\,\,\big|\,\,\, \mathcal{B}(\mathbf{x}) \geq \gamma_b\Big\}.\\
	\end{align*}
	Then, by applying Theorem 3 in \citep{Kushner} and employing conditions in  \eqref{subeq:decreasing} and \eqref{subeq:initial}, respectively, one has
	\begin{align*}
		&  \PP\Big\{ {\mathbf x}_{{\mathbf  x_0uw}}(k) \!\in\!  X_b \!\!\times\!\! (X\backslash X_b)^h \text { for some } {k\! \in\! \{1,\dots,\!\mathcal{T}\}} \,\big|\, \mathbf{x}_0 \Big\} \\
		& \quad \leq  \PP \Big\{\sup _{{k \in \{1,\dots,\mathcal{T}\}}} \mathcal{B}\big({\mathbf x}_{{\mathbf  x_0uw}}(k)\big) \geq \gamma_b  \,\big|\, \mathbf{x}_0 \Big\} \leq \frac{\gamma_a + \eta \mathcal{T}}{\gamma_b}  \!.
	\end{align*}
	Taking the complement yields
	\begin{align}\notag
		\PP&\Big\{ {\mathbf x}_{{\mathbf  x_0uw}}(k) \!\notin\!  X_b \!\times\! (X\backslash X_b)^h \text { for all } {k\! \in\! \{1,\dots,\!\mathcal{T}\}} \,\big|\, \mathbf{x}_0 \Big\}\\\label{eq:safety_prob} &~ \ge 1 - \dfrac{\gamma_a + \eta \mathcal{T}}{\gamma_b}.
	\end{align}
	which completes the proof. $\hfill\blacksquare$
\begin{remark}
	For a system with a time-invariant delay $h$, the probability of transitioning from the current state $x_k$ to the next state $x_{k+1}$ depends not only on the current state but also on the preceding $h$ states. This dependence characterizes the system as an order-$h$ Markov process. By augmenting the state space to include the delay history, i.e., representing it as an $(h+1)$-tuple $\mathbf{x}_k =(x_k,x_{k-1},\ldots,x_{k-h}) \in X^{h+1}$ as in Proposition~\ref{Th:safety}, the process can be equivalently viewed as an order-$1$ Markov process under condition~\eqref{subeq:decreasing}. Consequently, the results of \citep[Theorem~3]{Kushner}, used in the final stage of the proof, are applicable to this setting.
\end{remark}
Satisfaction of the proposed conditions outlined in Definition~\ref{def:KCBC} enables the design of a safety controller that provides a probabilistic safety guarantee, as established in Proposition~\ref{Th:safety}. While promising, the condition in~\eqref{subeq:decreasing} presents the main challenge, as it captures the dynamics of the dt-SNPS-td with time-invariant delays. {To address this, we exploit the {quadratic} structure of the Krasovskii control barrier certificates to reformulate~\eqref{subeq:decreasing} into a set of matrix inequalities.}

In the following, the K-QCBC method is proposed for two cases: with and without input constraints of \eqref{input_set}, as described in Subsections \ref{with_input} and \ref{without_input}, respectively. For K-QCBC with input constraints, the problem is {non-convex} since we aim to jointly design K-QCBC and a safety controller while imposing the input constraints specified in \eqref{input_set}. To address this non-convexity, we propose an approach that resolves it in a conservative manner while enabling the joint design of the K-QCBC and a safety controller under these constraints, as demonstrated in the results of Theorem~\ref{Th:decay}. Despite the underlying conservatism, the effectiveness of the K-QCBC with input constraints has been validated on three complex case studies (cf.~Section~\ref{sec:Case}). In Subsection \ref{without_input}, we then present results for K-QCBC without input constraints, leveraging a new condition that eliminates the aforementioned conservatism (cf. Proposition~\ref{new87}). The results in Subsections~\ref{with_input} and~\ref{without_input} highlight a trade-off between enforcing input constraints and the conservativeness of the underlying safety conditions.

\subsection{K-QCBC Design with Input Constraints}\label{with_input}

In this subsection, leveraging the K-QCBC, we reformulate~\eqref{subeq:decreasing} as a set of matrix inequalities that enforce the input constraints in~\eqref{input_set}, while also fulfilling~\eqref{subeq:initial}–\eqref{subeq:unsafe}. The following theorem formally presents these results, which constitute the first main contribution of this work.
\begin{theorem}[\textbf{K-QCBC with Input Constraints}]\label{Th:decay}
	Consider a dt-SNPS-td $\Sigma = \left({\mathcal A, \mathcal A_1,  \mathcal{G},\allowbreak E, X, \allowbreak U, h}\right)\!,$ with the input constraint~\eqref{input_set}. Suppose there exist matrices $P \succ 0$, $P_1 \succeq 0$, symmetric matrix $\mathcal{S}$, state-dependent polynomial matrices $\mathcal{F}$, $\mathcal{F}_1$ and constant $\alpha \in \mathbb{R}^+$ such that following conditions hold:
	\begin{subequations}\label{K-PCBC}
		\begin{align}
			\label{conin}
			&x^\top Px + h \bar{x}^\top P_1\bar{x} \leq \gamma_a, ~~\quad\quad\quad \forall (x,\bar{x}) \in {X_a^2},
			\\\label{conun}
			&x^\top Px \geq \gamma_b,   \quad\quad\quad\quad\quad\quad\quad\quad\,\, \forall x \in {X_b},
			\\\label{con1}
			&\begin{bmatrix}
				P-P_1  &  \mathbf{0}& (\mathcal A  +  \mathcal G \mathcal{F} )^\top  \\
				\star  &   P_1 & (\mathcal A_1 + \mathcal G \mathcal{F}_1 )^\top  \\
				\star   &  \star   & \mathcal{S}
			\end{bmatrix} \succeq 0,   ~~\;\; \forall (x,x_{h}) \in {X}^2,\\\label{con2}
			&P \preceq \frac{1}{\alpha} \mathbf{I}_{n},\\\label{con3}
			&\mathcal{S} \preceq {\alpha} \mathbf{I}_{n},\\\label{con4}
			&1 \!-\! b_j^{\top} (\mathcal{F}x \!+\! \mathcal{F}_1x_h) \geq 0,\;\; \forall (x, x_{h} ) \in {X}^2, \, j=1, \ldots, J.
		\end{align}
	\end{subequations}
	Then, conditions~\eqref{subeq:initial}-\eqref{subeq:decreasing} together with the input constraint~\eqref{input_set} are all enforced {with the control input and CBC as in \eqref{feedback} and \eqref{k-CBC}, respectively}.
\end{theorem}
{\bf Proof.}
	We begin by showing that the satisfaction of \eqref{conin} and \eqref{conun} ensures the satisfaction of the conditions in \eqref{subeq:initial} and \eqref{subeq:unsafe}. Given that $\mathcal{B}(\mathbf{x})=x^\top P x + \sum_{i=1}^{h} x^\top_{i} P_1 x_{i}$, one has
	\begin{align*}
		\sup_{\mathbf{x} \in {X}_a^{h+1}} \mathcal{B}(\mathbf{x}) 
		&= \sup_{x \in {X}_a} x^\top P x +  \sum_{i=1}^{h} \sup_{x_i \in {X}_a} x_i^\top P_1 x_i \\
		&~= \sup_{x \in {X}_a} x^\top P x + h \sup_{\bar{x} \in {X}_a} \bar{x}^\top P_1 \bar{x} \\
		&~= \sup_{(x,\bar{x}) \in {X}_a^2} x^\top P x + h  \bar{x}^\top P_1 \bar{x}.
	\end{align*}
	Hence, satisfying condition \eqref{conin} guarantees that \eqref{subeq:initial} is met. Additionally, one has
	\begin{align*}
		\inf_{\mathbf{x} \in {X}_b \times ({X} \backslash {X}_b)^h} \!\! \mathcal{B}(\mathbf{x}) 
		&=\! \inf_{x \in {X}_b} x^\top P x +\! \sum_{i=1}^{h}  \inf_{x_i \in ({X} \backslash {X}_b)} \!\! x_i^\top P_1 x_i \\
		&\stackrel{P_1 \succeq 0}{\geq} \inf_{x \in {X}_b} x^\top P x.
	\end{align*}
	Thus, satisfying condition \eqref{conun} ensures that \eqref{subeq:unsafe} is also fulfilled.
	
	\noindent We proceed by showing the satisfaction of~\eqref{subeq:decreasing}, which represents the most challenging condition. Considering
	$\mathbf{x}_{k+1}$ in~\eqref{successor}, we have
	\begin{align}\notag
		&\EE\big[ \mathcal B(\mathbf{x}_{k+1}) \,\,\,\big|\,\,\, \mathbf{x}_k,u_k\big]
		\\\notag&~~~~=\EE\big[x^{\top}_{k+1} P x_{k+1} + \sum_{i=1}^h x^{\top}_{k-i+1} P_1 x_{k-i+1}  \,\,\,\big|\,\,\, \mathbf{x}_k,u_k\big] 
		\\\notag&~~~~= \EE\big[x^{\top}_{k+1} P x_{k+1}\,\,\,\big|\,\,\, \mathbf{x}_k,u_k\big] \\\label{telescopic} &\qquad+ \EE\big[\sum_{i=1}^h x^{\top}_{k-i+1} P_1 x_{k-i+1} \,\,\,\big|\,\,\, \mathbf{x}_k,u_k\big].
	\end{align}
	According to~\eqref{history}, and given the conditional expectation with respect to $\mathbf{x}_k$ and $u_k$, the expectation operator can be omitted since each state is conditioned on itself:
	\begin{equation}\label{telescopic2}
		\EE\big[\sum_{i=1}^h \! x^{\top}_{k-i+1} P_1 x_{k-i+1} \,\,\,\big|\,\,\, \mathbf{x}_k,u_k\big] \!=\! \sum_{i=1}^h \! x^{\top}_{k-i+1} P_1 x_{k-i+1}.
	\end{equation}
	From \eqref{telescopic2}, we obtain the equivalent form of \eqref{telescopic} as
	\begin{align}\notag
		&\EE\big[ \mathcal B(\mathbf{x}_{k+1}) \,\,\,\big|\,\,\, \mathbf{x}_k,u_k\big]\\\notag
		&~~~~=\EE\big[x^{\top}_{k+1} P x_{k+1}  \,\,\,\big|\,\,\, \mathbf{x}_k,u_k\big] 
		+ \sum_{i=1}^h x^{\top}_{k-i+1} P_1 x_{k-i+1} \\\notag
		&~~~~=\EE\big[x^{\top}_{k+1} P x_{k+1}  \,\,\,\big|\,\,\, \mathbf{x}_k,u_k \big] 
		+x^{\top}_k P_1 x_k\\\label{telescopic3} & \qquad
		+\sum_{i=2}^h x^{\top}_{k-i+1} P_1 x_{k-i+1}.
	\end{align}
	Given $\mathbf{x}_k$ in \eqref{history}, by subtracting $\mathcal{B}(\mathbf{x}_k)$ from both sides of~\eqref{telescopic3} and re-indexing the summation over delayed states, we have
	\begin{align*}
		&\EE\big[ \mathcal B(\mathbf{x}_{k+1}) \,\,\big|\,\, \mathbf{x}_k,u_k\big] - \mathcal B(\mathbf{x}_k)  \\ 
		&~~=\EE\big[x^{\top}_{k+1} P x_{k+1}  \,\,\big|\,\, \mathbf{x}_k,u_k \big] 
		\!+ \! x^{\top}_k P_1 x_k
		\!+\! \underbrace{\sum_{i=1}^{h-1} \! x^{\top}_{k-i} P_1 x_{k-i}}_{\text{re-indexed}}  \\[-5ex] &~~-\big(\overbrace{x^{\top}_k P x_k
			+\sum_{i=1}^{h} \! x^{\top}_{k-i} P_1 x_{k-i}\big)}^{\mathcal B(\mathbf{x}_k)}.
	\end{align*}
	By cancellation of the overlapping terms in the sums over \( P_1 \), we simplify the series and obtain
	\begin{align}\notag
		&\EE\big[ \mathcal B(\mathbf{x}_{k+1})  \,\,\,\big|\,\,\, \mathbf{x}_k,u_k\big] - \mathcal B(\mathbf{x}_k)  \\\notag  
		&~~~~= \EE\big[x^{\top}_{k+1} P x_{k+1}  \,\,\,\big|\,\,\, \mathbf{x}_k,u_k\big] 
		+ x^{\top}_{k} P_1 x_{k} - x^{\top}_{k} P x_{k} \\\label{telescopic4} & \qquad - x^{\top}_{k-h} P_1 x_{k-h}.
	\end{align}
	From \eqref{history-simple}, we set \( x_{k-h} = x_h \) as the state delayed by \( h \) steps and \( x_k = x \) as the current state. By applying the conditional expectation $\EE\big[x^{\top}_{k+1} P x_{k+1}  \,\,\,\big|\,\,\, \mathbf{x}_k=\mathbf{x}, u_k=u\big]$ in~\eqref{telescopic4} to the dt-SNPS-td $\Sigma$ in~\eqref{eq:dt-SNPS-td} with the state-feedback controller $u=\mathcal{F}x + \mathcal{F}_1x_h$ in~\eqref{feedback}, and since {$w\sim\mathcal{N}(\mathbf{0},\mathbf{I}_n)$}, we obtain
	\begin{align*}
		&\EE\big[ \mathcal B(\mathbf{x}_{k+1})  \,\,\,\big|\,\,\, \mathbf{x}_k=\mathbf{x},u_k=u\big] - \mathcal B(\mathbf{x}_k)\\ &=
		\EE \big[( \mathcal Ax  + \mathcal A_1 x_h  +   \mathcal G u+ Ew )^\top P ( \mathcal Ax  + \mathcal A_1 x_h  +   \mathcal G u\\ & ~~~~ + Ew ) \,\,\,\big|\,\,\, \mathbf{x},u\big]- x^\top P x +x^\top P_1 x - x_{h}^\top P_1 x_{h}\\ &=
		\EE \big[x^\top( \mathcal A  +  \mathcal G \mathcal{F} )^\top P ( \mathcal A   +  \mathcal G \mathcal{F} )  x\,\,\,\big|\,\,\, \mathbf{x},u\big] \\ &~~~~+  \EE\big[x_{h}^\top  (\mathcal A_1 + \mathcal G \mathcal{F}_1 )^\top  P (\mathcal A_1 + \mathcal G \mathcal{F}_1 )  x_{h}\,\,\,\big|\,\,\, \mathbf{x},u\big] \\ &~~~~ + 2\EE\big[x_{h}^\top(\mathcal A_1 + \mathcal G \mathcal{F}_1 )^\top P ( \mathcal A   +  \mathcal G \mathcal{F} )  x \,\,\,\big|\,\,\, \mathbf{x},u\big]\\
		&~~~~+ \EE\big[w^\top E^\top P E w\,\big|\, \mathbf{x},u\big] - x^\top P x +x^\top P_1 x - x_{h}^\top P_1 x_{h}.  
	\end{align*}
	By applying the expected values, collecting terms in  $x$ and  $x_{h}$, and using the cyclic property of the trace, $\EE\big[w^\top E^\top P E w\big] = \mathsf{Tr}(\EE\big[w^\top E^\top \allowbreak P E w\big]) = \mathsf{Tr}(E^\top P E\,\EE\big[ww^\top \big]) $, one has 
	\begin{align*}
		&\EE\big[ \mathcal B(\mathbf{x}_{k+1})  \,\,\,\big|\,\,\, \mathbf{x}_k=\mathbf{x},u_k=u\big] - \mathcal B(\mathbf{x}_k) \\ &=
		x^\top  \Big(( \mathcal A  +  \mathcal G \mathcal{F} )^\top P ( \mathcal A   +  \mathcal G \mathcal{F} ) -P +P_{1}  \Big) x\\ &\,\,\,\,\,\, +  x_{h}^\top  \Big((\mathcal A_1 + \mathcal G \mathcal{F}_1 )^\top  P (\mathcal A_1 + \mathcal G \mathcal{F}_1 ) - P_1 \Big)  x_{h}\\&\,\,\,\,\,\, + 2 x_{h}^\top \Big( (\mathcal A_1 + \mathcal G \mathcal{F}_1 )^\top P ( \mathcal A   +  \mathcal G \mathcal{F} )\Big)  x\\ &\,\,\,\,\,\, + \mathsf{Tr}(E^\top P E\underbrace{\EE[ww^\top]}_{{\mathbf{I}_n}}).
	\end{align*}
	Then, to satisfy \eqref{subeq:decreasing}, it is sufficient to show that
	\begin{align}\label{step1}
		\begin{bmatrix}
			x \\
			x_{h}
		\end{bmatrix}^\top \begin{bmatrix}
			\Lambda_{11} & \Lambda_{12}  \\
			\star &  \Lambda_{22}
		\end{bmatrix} \begin{bmatrix}
			x \\
			x_{h}
		\end{bmatrix} \leq 0,
	\end{align}
	with {$\eta = \mathsf{Tr}(E^\top P E)$}, and
	\begin{align*}
		\Lambda_{11} &= (\mathcal A  +  \mathcal G \mathcal{F} )^\top P ( \mathcal A   +  \mathcal G \mathcal{F} ) -P +P_{1},\\
		\Lambda_{12} &=  ( \mathcal A   +  \mathcal G \mathcal{F} )^\top P (\mathcal A_1 + \mathcal G \mathcal{F}_1 ), \\
		\Lambda_{22} &=  (\mathcal A_1 + \mathcal G \mathcal{F}_1 )^\top  P (\mathcal A_1 + \mathcal G \mathcal{F}_1 )  - P_1.
	\end{align*}
	According to Schur complement \citep{zhang2006schur}, condition in~\eqref{step1} is satisfied if
	\begin{align}\label{step3}
	\mathbf{R} \!=\!	\begin{bmatrix}
		P -P_1 & \mathbf{0}& (\mathcal A  +  \mathcal G \mathcal{F} )^\top  \\
		\star &  P_1 & (\mathcal A_1 + \mathcal G \mathcal{F}_1 )^\top \\
		\star &  \star  &  {P}^{-1}
	\end{bmatrix} \succeq 0.
\end{align}
{The presence of both \(P\) and \(P^{-1}\) as decision variables in \eqref{step3} leads to a non-convex formulation. To resolve this, \eqref{step3} can be reformulated as}
\begin{align}\label{step4}
	\mathbf{R} \!=\!	\underbrace{\begin{bmatrix}
			P-P_1 & \mathbf{0}& (\mathcal A  +  \mathcal G \mathcal{F} )^\top  \\
			\star & P_1 & (\mathcal A_1 + \mathcal G \mathcal{F}_1 )^\top  \\
			\star  & \star  & \mathcal{S}
	\end{bmatrix}}_{\mathbf{R}_1} + \underbrace{\begin{bmatrix}
			\mathbf{0} & \mathbf{0}& \mathbf{0}  \\
			\mathbf{0}   & \mathbf{0}  & \mathbf{0}  \\
			\mathbf{0}   & \mathbf{0}   & \underbrace{{P}^{-1} - \mathcal{S}}_{\mathbf{H}}
	\end{bmatrix}}_{\mathbf{R}_2} \succeq 0,
\end{align}
	for any $n \times n$ symmetric matrix $\mathcal{S}$. Since $P\succ0$ and $P=P^\top$, conditions \eqref{con2}--\eqref{con3} yield $P^{-1}\succeq \alpha \mathbf{I}_n \succeq \mathcal{S}$, and accordingly $\mathbf{H}\succeq 0$ and $\mathbf{R}_2\succeq 0$. Moreover, \eqref{con1} implies $\mathbf{R}_1\succeq 0$, so $\mathbf{R}\succeq 0$. Hence, it follows that \eqref{subeq:decreasing} holds under the input constraint \eqref{input_set}, which is equivalent to~ condition~\eqref{con4}, thereby completing the proof. $\hfill\blacksquare$
\begin{remark}
It is worth highlighting that the safety conditions \eqref{subeq:initial}-\eqref{subeq:decreasing} in Definition~\ref{def:KCBC} should be satisfied over the $(h+1)$-fold Cartesian product of the relevant sets, the state set, initial set, and unsafe set, accounting for the $h$-step delay. In contrast, the safety conditions proposed using the K-QCBC in Theorem~\ref{Th:decay} (and also the K-PCBC in Lemma~\ref{sos}) involve only a 2-fold Cartesian product of these sets. This refinement significantly reduces the computational complexity of safety verification for time-delayed systems, representing a key novelty of the proposed framework.
\end{remark}

\begin{remark}
We note that $\alpha$ in \eqref{con2},\eqref{con3} is introduced to avoid the simultaneous appearance of $P$ and $P^{-1}$ in \eqref{step4} when enforcing input constraints, which would otherwise result in a non-convex formulation. Specifically, $\alpha$ and $\alpha^{-1}$ appear in \eqref{con2} and \eqref{con3} are fixed a priori to eliminate such non-convexity (cf. \eqref{con1}). This may require checking or adjusting $\alpha$ if the conditions are not satisfied, representing a trade-off between tractability and conservatism. In the absence of input constraints, $\alpha$ does not appear (cf. \eqref{con1-R}).
\end{remark}

\subsection{Computation of K-QCBC and Safety Controller}\label{sec:Computation}
To compute the K-QCBC and its corresponding safety controller proposed in Theorem~
\ref{Th:decay}, we present the following lemma, which transforms conditions \eqref{conin}-\eqref{con1} and \eqref{con4} into an SOS optimization program.

\begin{lemma}[\textbf{K-QCBC: SOS Program}]\label{SOS_K-QCBC} 
	Consider the state set $X$, the initial set $X_a$, and the unsafe set $X_b$, which are outlined by vectors of polynomial inequalities as $ {X}=\left\{ x  \in \mathbb{R}^n  \,\,\,\big|\,\,\, \mathcal{J}(x) \geq 0\right\}$, $X_{a}=\left\{x \in \mathbb{R}^n  \,\,\,\big|\,\,\, \mathcal{J}_{a}(x) \geq\right.$ $0\}$, and $ X_{b}=\left\{x \in \mathbb{R}^n  \,\,\,\big|\,\,\, \mathcal{J}_{b}(x) \geq 0\right\}$. Then, $\mathcal B(\mathbf x) =x^\top P x + \sum_{i=1}^{h} x^\top_{i} P_1 x_{i}$ is a K-QCBC for dt-SNPS-td in~\eqref{eq:dt-SNPS-td}, and $u=\mathcal{F}x + \mathcal{F}_1 x_h$ in~\eqref{feedback} is its safety controller if there exist matrices $P \succ 0$, $P_1 \succeq 0$, symmetric matrix $\mathcal{S}$, constants $ \alpha, \gamma_a, \gamma_b  \in \mathbb{R}^{+}$, with $\gamma_b>\gamma_a$, and vectors of SOS polynomials, $\hat{\mathcal{Y}}_j(x,x_h)$, $\hat{\mathcal{Y}}_{h_j}(x,x_h)$, $j\in\{1,\dots,J\}$, $
	\mathcal{Y}(x,x_h),\mathcal{Y}_{h}(x,x_{h})$, $\mathcal{Y}_{a}(x), \bar{\mathcal{Y}_{a}}(\bar{x})$, and $
	\mathcal{Y}_{b}(x)$, such that the following expressions are all SOS polynomials:
	\begin{subequations}\label{LEMMA1}
	\begin{align}\label{L1} 
		-&x^\top P x - h \bar{x}^\top P_1 \bar{x} - \mathcal{Y}_{a}^{\top}(x) \mathcal{J}_{a}(x) - \bar{\mathcal{Y}_{a}}^{\top}(\bar{x}) \mathcal{J}_{a}(\bar{x}) +\gamma_a, \\\label{L2}
		&x^\top P x  -\mathcal{Y}_{b}^{\top}(x) \mathcal{J}_{b}(x) - \gamma_b, \\\notag 
		&\begin{bmatrix}
			P-P_1 & \mathbf{0}& (\mathcal A  +  \mathcal G \mathcal{F} )^\top  \\
			\star&  P_1 & (\mathcal A_1 + \mathcal G \mathcal{F}_1 )^\top  \\
			\star&  \star & \mathcal{S}
		\end{bmatrix}\\\label{L3} &~~~ {- \big( \mathcal{Y}(x,x_h )^\top\mathcal{J}(x) +\ \mathcal{Y}_h(x,x_{h} )^\top  \mathcal{J}(x_{h}) \big)\mathbf{I}_{3n}},\\\notag
		& 1 \!-\! b_j^{\top}(\mathcal{F}x \!+\! \mathcal{F}_1x_h) \!-\! {\hat{\mathcal{Y}_j}(x,x_h )^\top\!\!\mathcal{J}(x)\!-\!\hat{\mathcal{Y}}_{h_j}(x,x_{h} )^\top \!\! \mathcal{J}(x_{h})},\\\label{L4} &j=1, \ldots,J\!.
	\end{align}
\end{subequations}
\end{lemma}
{\bf Proof.}
	Since $\mathcal{Y}_{a}(x),\bar{\mathcal{Y}}_{a}(\bar{x})$ are SOS polynomials, it follows that $\mathcal{Y}_{a}^{\top}(x) \mathcal{J}_{a}(x) + \bar{\mathcal{Y}_{a}}^{\top}(\bar{x}) \mathcal{J}_{a}(\bar{x}) \ge 0$ for all $(x,\bar{x}) \in X^{2}_a$. Therefore, feasibility of \eqref{L1} ensures \eqref{conin}. By similar reasoning, \eqref{L2} implies \eqref{conun}. We now show that \eqref{L3} and \eqref{L4} imply \eqref{con1} and \eqref{con4}. As $\mathcal{Y}(x,x_h)$, $\mathcal{Y}_h(x,x_{h})$, {$\hat{\mathcal{Y}}_j(x,x_h)$, $\hat{\mathcal{Y}}_{h_j}(x,x_{h})$, $j\in\{1,\dots,J\}$} are SOS polynomials, {we have $\mathcal{Y}^\top(x,x_h)\mathcal{J}(x) +\mathcal{Y}_h^\top(x,x_h)\mathcal{J}(x_h)\ge 0$}, and {$\hat{\mathcal{Y}_j}(x,x_h )^\top\!\!\mathcal{J}(x) +\hat{\mathcal{Y}}_{h_j}(x,x_{h} )^\top \!\! \mathcal{J}(x_{h})\geq0, j\in\{1,\dots,J\}$} for all $(x,x_h) \in {X}^2$. Given that \eqref{L3} and \eqref{L4} are also SOS polynomials, then
		\begin{align*}
		&\begin{bmatrix}
			P-P_1 & \mathbf{0}& (\mathcal A  +  \mathcal G \mathcal{F} )^\top  \\
			\star &  P_1 & (\mathcal A_1 + \mathcal G \mathcal{F}_1 )^\top  \\
			\star &  \star & \mathcal{S}
		\end{bmatrix}\\ &~~~
		- \Big(\mathcal{Y}^\top(x,x_h)\,\mathcal{J}(x) + \mathcal{Y}_h^\top(x,x_{h})\,\mathcal{J}(x_{h})\Big)\mathbf{I}_{3n} \succeq 0,\\[2mm]
		& 1 \!-\! b_j^{\top}(\mathcal{F}x + \mathcal{F}_1x_h) - \hat{\mathcal{Y}}_j(x,x_h )^\top\!\!\mathcal{J}(x)\\ &~ - \hat{\mathcal{Y}}_{h_j}(x,x_{h} )^\top \!\! \mathcal{J}(x_{h}) \geq 0,
	\end{align*}
	$j=1, \ldots, J$, therefore \eqref{con1} and \eqref{con4} hold, which in turn yields \eqref{subeq:decreasing}. $\hfill\blacksquare$

\subsection{K-QCBC Design without Input Constraints}\label{without_input}
While the proposed conditions in~\eqref{K-PCBC} offer the advantage of incorporating input constraints, they introduce a degree of conservatism, as the aggregate condition $\mathbf{R}_1 + \mathbf{R}_2 = \mathbf{R} \succeq 0$ in~\eqref{step4} is replaced by the separate constraints $\mathbf{R}_1 \succeq 0$ and $\mathbf{R}_2 \succeq 0$ in~\eqref{con1}–\eqref{con3}. Moreover, these conditions may exhibit sensitivity to the parameter $\alpha$, primarily because of how the input constraint~\eqref{input_set} is incorporated into the formulation. We next present a proposition that provides a relaxed condition applicable to the setting without input constraints, \emph{i.e.,} $u_k \in {U} \subseteq \mathbb{R}^{m}$.

\begin{proposition}[{\textbf{K-QCBC without Input~\eqref{input_set}}}]\label{new87}
	{Consider a dt-SNPS-td $\Sigma = {(\mathcal A, \mathcal A_1,  \mathcal{G},\allowbreak E, X, \allowbreak U, h)},$ with $U \in \mathbb{R}^{m}$.} Suppose there exist matrices $P \succ 0$, $\tilde{P}_1 \succeq 0$, and state-dependent polynomial matrices $\mathcal{Z}(x)$, $\mathcal{Z}_1(x_h)$ such that following condition holds:
\end{proposition}
\begin{align}\label{con1-R}
	&\begin{bmatrix}
		P^{-1} - \tilde{P}_1 &\mathbf{0}& P^{-1} \mathcal A^\top  +  \mathcal Z^\top  \mathcal G^\top  \\
		\star &  \tilde{P}_1 & P^{-1} \mathcal A^\top_1 + \mathcal Z^\top_1  \mathcal G^\top  \\
		\star & \star &  P^{-1}
	\end{bmatrix} \succeq 0, \;\; \forall (x, x_{h} ) \in X^2.
\end{align}
Then, the condition in~\eqref{subeq:decreasing} is enforced.

{\bf Proof.}
	According to \eqref{telescopic4}, in the first stage of the proof of Theorem~\ref{Th:decay} one has
	\begin{align*}
		&\EE\big[ \mathcal B(\mathbf{x}_{k+1})  \,\,\,\big|\,\,\, \mathbf{x}_k,u_k\big] - \mathcal B(\mathbf{x}_k)  \\  
		&~~~~= \EE\big[x^{\top}_{k+1} P x_{k+1}  \,\,\,\big|\,\,\, \mathbf{x}_k,u_k\big] 
		+ x^{\top}_{k} P_1 x_{k} - x^{\top}_{k} P x_{k} \\ &\qquad - x^{\top}_{k-h} P_1 x_{k-h}.
	\end{align*}
	From \eqref{history-simple}, we set \( x_{k-h} = x_h \) as the state delayed by \( h \) steps and \( x_k = x \) as the current state. Let us rewrite $\mathcal{B}(\mathbf{x})$ as $\mathcal{B}(\mathbf{x})=x^\top P^{\top} P^{-1} P x + \sum_{i=1}^{h} x_i^\top P\tilde{P}_1 P x_i$ with $P=P^{\top}$,  $P^{-1}P=\mathbf{I}_n$, and ${P}_1 = P\tilde{P}_1P$.
	By applying the conditional expectation to the dt-SNPS-td $\Sigma$ with the state-feedback controller $u=\mathcal{F}x + \mathcal{F}_1x_h$, and since {$w\sim\mathcal{N}(\mathbf{0},\mathbf{I}_n)$}, we obtain
	\begin{align*}
		&\EE\big[ \mathcal B(\mathbf{x}_{k+1})  \,\,\,\big|\,\,\, \mathbf{x}_k=\mathbf{x},u_k=u\big] - \mathcal B(\mathbf{x}_k)  \\ &=
		\EE \big[( \mathcal Ax  + \mathcal A_1 x_h  +   \mathcal G u+ Ew )^\top P ( \mathcal Ax  + \mathcal A_1 x_h  +   \mathcal G u\\ & ~~~ + Ew ) \,\,\,\big|\,\,\, \mathbf{x},u\big] {- x^\top P x +x^\top P_1 x - x_{h}^\top P_1 x_{h}}\\ &=
		x^\top P  \Big( P^{-1}( \mathcal A  \!+\!  \mathcal G \mathcal{F} )^\top P ( \mathcal A   \! +\!  \mathcal G \mathcal{F} )P^{-1} \!-\! P^{-1} +\tilde{P}_{1}  \Big) P x\\ &~~~ +  x_{h}^\top P  \Big(P^{-1}(\mathcal A_1 \!+\! \mathcal G \mathcal{F}_1 )^\top  P (\mathcal A_1 \!+\! \mathcal G \mathcal{F}_1 )P^{-1} \!-\! \tilde{P}_1 \Big)  P x_{h}\\&~~~  + 2 x_{h}^\top P \Big(P^{-1} (\mathcal A_1 + \mathcal G \mathcal{F}_1 )^\top P ( \mathcal A   +  \mathcal G \mathcal{F} )P^{-1}\Big)  P x \\ &~~~+ \mathsf{Tr}(E^\top P E\underbrace{\EE[w  w^\top]}_{{\mathbf{I}_n}}).
	\end{align*}
	Then, to satisfy $\big[ \mathcal B(\mathbf{x}_{k+1})  \,\,\,\allowbreak \big|\,\,\, \mathbf{x}_k=\mathbf{x},u_k=u\big] - \mathcal B({\mathbf{x}}) \leq  \eta$ in~\eqref{subeq:decreasing}, it is sufficient to show
	\begin{align}\label{step1-R}
		\begin{bmatrix}
			Px \\
			P x_{h}
		\end{bmatrix}^\top \begin{bmatrix}
			\Gamma_{11} & 	\Gamma_{12}  \\
			\star &  	\Gamma_{22}
		\end{bmatrix} \begin{bmatrix}
			P x \\
			P x_{h}
		\end{bmatrix} \leq 0,
	\end{align}
	with {$\eta = \mathsf{Tr}(E^\top P E)$}, and
	\begin{align*}
		\Gamma_{11} &= P^{-1}(\mathcal A  +  \mathcal G \mathcal{F} )^\top P ( \mathcal A   +  \mathcal G \mathcal{F} )P^{-1} -P^{-1} +\tilde{P}_{1},\\
		\Gamma_{12} &= {P^{-1}(\mathcal A + \mathcal G \mathcal{F} )^\top P ( \mathcal A_1   +  \mathcal G \mathcal{F}_1 )P^{-1},} \\
		\Gamma_{22} &=  P^{-1}(\mathcal A_1 + \mathcal G \mathcal{F}_1 )^\top  P (\mathcal A_1 + \mathcal G \mathcal{F}_1 )P^{-1}  - \tilde{P}_1.
	\end{align*}
	By setting $\mathcal{Z}= \mathcal{F}P^{-1}$, $\mathcal{Z}_1= \mathcal{F}_1P^{-1}$, and since ${P}_1 = P\tilde{P}_1P$, according to Schur complement~\citep{zhang2006schur}, condition in~\eqref{step1-R} is satisfied if 
	\begin{align}\label{step3-R}
		\begin{bmatrix}
			P^{-1} - \tilde{P}_1 & \mathbf{0}& P^{-1} \mathcal A^\top  +  \mathcal Z^\top  \mathcal G^\top  \\
			\star &  \tilde{P}_1 & P^{-1} \mathcal A^\top_1 + \mathcal Z^\top_1  \mathcal G^\top  \\
			\star & \star  &  P^{-1}
		\end{bmatrix} \succeq\! 0.
	\end{align}
	As this is condition~\eqref{con1-R}, which is enforced, it follows that~\eqref{subeq:decreasing} holds, thereby concluding the proof.$\hfill\blacksquare$

\begin{remark}
	One can consider \( P^{-1} = \mathcal{C} \) in \eqref{con1-R} and solve for \( \mathcal C \) such that it is symmetric and positive definite. Once \( \mathcal{C}  \) is determined, its inverse will yield the matrix \( P \).
\end{remark}

To compute the K-QCBC and its corresponding safety controllers without input constraints, condition \eqref{con1-R} can be expressed as an SOS condition as the following lemma.
\begin{lemma}[\textbf{SOS Program for \eqref{con1-R}}]
	Consider the state set $X$, the initial set $X_a $, and the unsafe set $X_b$, each of which is outlined by vectors of polynomial inequalities as $ {X}=\left\{ x  \in \mathbb{R}^n  \,\,\,\big|\,\,\, \mathcal{J}(x) \geq 0\right\}$, $X_{a}=\left\{x \in \mathbb{R}^n  \,\,\,\big|\,\,\, \mathcal{J}_{a}(x) \geq 0\right\}$, and $ X_{b}=\left\{x \in \mathbb{R}^n  \,\,\,\big|\,\,\, \mathcal{J}_{b}(x) \geq 0\right\}$, respectively. Then, condition~\eqref{con1-R} holds if there exist matrices $P \succ 0$, $\tilde{P}_1 \succeq 0$, state-dependent polynomial matrices $\mathcal{Z}$, $\mathcal{Z}_1$, and vectors of SOS polynomials, {$
		\mathcal{Y}(x,x_h)$, and  $\mathcal{Y}_{h}(x,x_{h})$ such that}
	\begin{align}
		&\begin{bmatrix}\notag
			P^{-1} \!-\! \tilde{P}_1 & \mathbf{0}& P^{-1} \mathcal A^\top  +  \mathcal Z^\top  \mathcal G^\top  \\
			\star &  \tilde{P}_1 & P^{-1} \mathcal A^\top_1 + \mathcal Z^\top_1  \mathcal G^\top  \\
			\star  & \star  &  P^{-1}
		\end{bmatrix}\\\label{L1-relaxed} &~~~- {\big( \mathcal{Y}(x,x_h )^\top\mathcal{J}(x) + \mathcal{Y}_h(x,x_{h} )^\top  \mathcal{J}(x_{h}) \big)\mathbf{I}_{3n}},
	\end{align}
	is an SOS polynomial.
\end{lemma}
{\bf Proof.}
	The proof follows a similar reasoning as that of Lemma~\ref{SOS_K-QCBC} and is omitted here due to space limitations.$\hfill\blacksquare$

\begin{remark}
	Given the quadratic structure of the K-QCBC, one can exploit this property to compute $\gamma_a$ and $\gamma_b$ as	\begin{subequations}\label{lemma-cons}
		\begin{align}\label{lemma:con1}
			& \gamma_a= (\lambda_{\max }({P})  + h\lambda_{\max }({P_1})) \max _{{x} \in {X}_{a}}\Vert {x}   \Vert^2, \\\label{lemma:con2}
			&\gamma_b=\lambda_{\min }({P}) \min _{x \in {X}_{b} }\Vert x   \Vert^2.
		\end{align}
	\end{subequations}
	This approach can be potentially less computationally demanding than the SOS-based conditions required for computing $\gamma_a$ and $\gamma_b$ in~\eqref{L1} and~\eqref{L2}.
\end{remark}

We present Algorithm~\ref{Alg1}, which outlines the steps for computing the K-QCBC and the corresponding safety controllers with input constraints.

\begin{algorithm}[t]
	\caption{Computation of K-QCBC and safety controllers}
	\label{Alg1}
	\begin{algorithmic}[1]
		\REQUIRE dt-SNPS-td $\Sigma$, $\Upsilon = (X_a, X_b, \mathcal T)$\\
		\STATE Select $\alpha$
		\IF{\eqref{con2}, \eqref{con3}, \eqref{L1}--\eqref{L4}, using \textsf{SeDuMi}~\citep{Sturm} and \textsf{SOSTOOLS}~\citep{prajna2004sostools}
			are satisfied}
		\RETURN $P$, $P_1$, $\gamma_a$, $\gamma_b$, $\mathcal{F}$, $\mathcal{F}_1$
		\ELSE
		\STATE Adjust $\alpha$
		\ENDIF\\
		\STATE Using  ${P}$, compute {$\eta = \mathsf{Tr}(E^\top P E)$}
		\STATE Compute the {safety guarantee $1-\mu_h = 1-\frac{\gamma_a + \eta \mathcal{T}}{\gamma_b}$}
		\ENSURE K-QCBC $\mathcal{B}(\mathbf x)=x^\top P x + \sum_{i=1}^{h} x^\top_{i} P_1 x_{i}$, probabilistic safety guarantee~in~\eqref{eq:safety_prob}, safety controller $u= \mathcal F x + \mathcal F_1x_{h}$, while enforcing the input constraint in~\eqref{input_set} 
	\end{algorithmic}
\end{algorithm}

\section{Krasovskii Polynomial Control Barrier Certificates}\label{subsec:problem_G}
In this section, we present a more general result in the sense of the selection of \emph{Krasovskii} control barrier certificates by allowing them to possess higher polynomial degrees. Specifically, we introduce the Krasovskii \emph{polynomial }control barrier certificates (K-PCBC) as
\begin{equation}\label{K-PCBC_temp} 
	\mathcal{B}(\mathbf{x})= g(x) + \sum_{i=1}^{h} \tilde{g}(x_i), \end{equation} 
where $g(x), \tilde{g}(x_i)$ are both polynomials. The K-PCBC can be defined analogously to Definition~\ref{def:KCBC}, with the same set of conditions, with the only distinction being that the barrier function is polynomial, as given in~\eqref{K-PCBC_temp}. Under conditions~\eqref{subeq:initial}–\eqref{subeq:decreasing} and with the K-PCBC defined in~\eqref{K-PCBC_temp}, one can invoke Proposition~\ref{Th:safety} to establish the safety of the dt-SNPS-td~$\Sigma$, ensuring robustness to time delays. 

To obtain the K-PCBC and its associated safety controllers, we state the following lemma that reformulates the corresponding conditions as an SOS optimization problem. 
\begin{lemma}[\textbf{K-PCBC: SOS Program}] \label{sos}
	Consider the state set $X$, the initial set $ X_a$, the unsafe set $X_b$, which are outlined by vectors of polynomial inequalities as $ {X}=\left\{ x  \in \mathbb{R}^n  \,\,\,\big|\,\,\, \mathcal{J}(x) \geq 0\right\}$, $X_{a}=\left\{x \in \mathbb{R}^n  \,\,\,\big|\,\,\, \mathcal{J}_{a}(x) \geq\right.$ $0\}$, $ X_{b}=\left\{x \in \mathbb{R}^n  \,\,\,\big|\,\,\, \mathcal{J}_{b}(x) \geq 0\right\}$, respectively, {and input set $U$ as presented in \eqref{input_set}.} 
	Let there exist an SOS polynomial $\mathcal{B}(\mathbf{x})= g(x) + \sum_{i=1}^{h} \tilde{g}(x_i)$, constants $\gamma_a,\gamma_b,\eta\in \mathbb{R}^+$, polynomial $\mathcal{Y}_{u_q}(x,x_h)$, corresponding to the $q^{\text {th }}$ input in  $u=[u_1; u_2; \ldots; u_{m}] \in U \subset \mathbb{R}^m$ and {$\mathcal{Y}_{u}(x,x_h)=[\mathcal{Y}_{u_1}(x,x_h); \mathcal{Y}_{u_2}(x,x_h); \ldots; \mathcal{Y}_{u_m}(x,x_h)]$}. Suppose there exist vectors of SOS polynomials
	{$\tilde{\mathcal{Y}}_j(x,x_h,u)$, $\mathcal{Y}_j(x,x_h)$, $\mathcal{Y}_{h_j}(x,x_{h})$}, $j=1,\cdots,J$, $\mathcal{Y}_{a}(x), \bar{\mathcal{Y}}_{a}(\bar{x})$, $\mathcal{Y}_b(x), \bar{\mathcal{Y}}_{b}(\bar{x})$, {$\mathcal{Y}(x,x_h,u)$, and $\mathcal{Y}_h(x,x_h,\allowbreak u)$} of appropriate dimensions such that the following expressions are all SOS polynomials:
\begin{subequations}\label{eq:sos}
	\begin{align}
		-&g(x) - h\, \tilde{g}(\bar{x})
		-  \mathcal{Y}_{a}^{\top}(x)\,\mathcal{J}_{a}(x)
		-  \bar{\mathcal{Y}}_{a}^{\top}(\bar{x})\,\mathcal{J}_{a}(\bar{x})
		+ \gamma_a, \label{eq:sos2}\\
		&g(x) + h\, \tilde{g}(\bar{x})
		- \mathcal{Y}_{b}^{\top}(x)\,\mathcal{J}_{b}(x)
		- \bar{\mathcal{Y}}_{b}^{\top}(\bar{x})\big(\mathcal{J}(\bar{x})-\mathcal{J}_b(\bar{x})\big)
		\notag\\
		&~~- \gamma_b, \label{eq:sos3}\\
		-&\EE\!\left[
		g(\mathcal Ax  \!+\! \mathcal A_1 x_h  \!+\!   \mathcal G u \!+\! Ew)
		\,\,\big|\,\,\mathbf{x},u\right]+ \eta 
		+ g(x) - \tilde{g}(x)\notag
	\end{align}
	\begin{align}
		&~~+ \tilde{g}\!\left(x_h\right) \!-\! \sum_{q=1}^{m}\big(u_q  \!-\! \mathcal{Y}_{u_q}(x,x_h) \big) 	\!-\! \mathcal{Y}^{\top}\!(x,x_h,u)\,\mathcal{J}(x)\notag\\
		&~~ - \mathcal{Y}^{\top}_h(x,x_h,u)\,\mathcal{J}(x_h) - \sum_{j=1}^{J}(1-b^\top_ju)\tilde{\mathcal{Y}}_j(x,x_h,u)\label{eq:sos5},\\
		&1\!-\!b^\top_j\mathcal{Y}_{u}(x,x_h)-\mathcal{Y}_j(x,x_h)\mathcal{J}(x)
		-\mathcal{Y}_{h_j}(x,x_{h})\mathcal{J}(x_h),\notag\\
		&\qquad j=1,\ldots,J. \label{eq:sos4}
	\end{align}
\end{subequations}
Then, $\mathcal{B}(\mathbf{x})= g(x) + \sum_{i=1}^{h} \tilde{g}(x_i)$, satisfying conditions~\eqref{subeq:initial}–\eqref{subeq:decreasing}, and {$u=[u_1; u_2; \ldots; u_{m}]=[\mathcal{Y}_{u_1}(x,x_h); \mathcal{Y}_{u_2}(x,x_h); \ldots; \mathcal{Y}_{u_m}(x,x_h)]$} is the corresponding safety controller.
\end{lemma}
{\bf Proof.}
	We begin by showing that the satisfaction of condition~\eqref{eq:sos2} implies condition~\eqref{subeq:initial}. Given that $\mathcal{B}(\mathbf{x})= g(x) + \sum_{i=1}^{h} \tilde{g}(x_i)$, one has
	\begin{align*}
		\sup_{\mathbf{x} \in {X}_a^{h+1}} \mathcal{B}(\mathbf{x}) 
		&=\! \sup_{x \in {X}_a} g(x) +  \sum_{i=1}^{h} \sup_{x_i \in {X}_a} \tilde{g}(x_i)\\ & = \sup_{x \in {X}_a} g(x) \!+\! h \sup_{\bar{x} \in {X}_a} \tilde{g}(\bar{x})
		\\
		&= \sup_{(x, \bar{x}) \in {X}_a^2} g(x) + h  \tilde{g}(\bar{x}).
	\end{align*} Since~\eqref{eq:sos2} is an SOS polynomial, then one has
	\begin{align*}
		g(x) + h  \tilde{g}(\bar{x})
		+  \mathcal{Y}_{a}^{\top}(x) \mathcal{J}_{a}(x) +  \bar{\mathcal{Y}}_{a}^{\top}(\bar{x}) \mathcal{J}_{a}(\bar{x})
		\leq \gamma_a.
	\end{align*}
	Since the term $\mathcal{Y}_{a}^{\top}(x) \mathcal{J}_{a}(x)+  \bar{\mathcal{Y}}_{a}^{\top}(\bar{x}) \mathcal{J}_{a}(\bar{x})$ in~\eqref{eq:sos2} is non-negative for all $(x,\bar{x}) \in X^2_a$, the SOS condition~\eqref{eq:sos2} without the nonnegative term $\mathcal{Y}_{a}^{\top}(x) \mathcal{J}_{a}(x)+  \bar{\mathcal{Y}}_{a}^{\top}(\bar{x}) \mathcal{J}_{a}(\bar{x})$ implies condition \eqref{subeq:initial}.
	Similarly, one has
	\begin{align*}
		\inf_{\mathbf{x} \in {X}_b \times ({X}\backslash {X}_b)^h } \mathcal{B}(\mathbf{x}) 
		&= \inf_{x \in {X}_b} g(x) +  \sum_{i=1}^{h} \inf_{x_i \in {X}\backslash  {X}_b} \tilde{g}(x_i) \\
		&~\geq \inf_{x \in {X}_b} g(x) + h \inf_{\bar{x} \in {X}\backslash {X}_b} \tilde{g}(\bar{x})\\
		&~= \inf_{(x,\bar{x}) \in {X}_b\times ({X}\backslash {X}_b)} g(x) + h \tilde{g}(\bar{x}).
	\end{align*}
	On the other hand, the terms $\mathcal{J}_{b}(x), (\mathcal{J}(\bar{x})-\mathcal{J}_b(\bar{x}))$ are non-negative for all $(x,\bar{x}) \in {X}_b\times({X}\backslash  {X}_b)$ due to the definition of the sets ${X}, {X}_b$. Therefore, condition~\eqref{eq:sos3} implies condition~\eqref{subeq:unsafe} due to the nonnegative term $\mathcal{Y}_{b}^{\top}(x) \mathcal{J}_{b}(x)+\bar{\mathcal{Y}_b}^{\top}(\bar{x}) (\mathcal{J}(\bar{x})-\mathcal{J}_b(\bar{x}))$. 
	
	\noindent We now show that condition~\eqref{eq:sos5} implies condition~\eqref{subeq:decreasing}, as well.
	Considering $\mathbf{x}_{k+1}$ as in~\eqref{successor}, by taking the expected value, one has
	\begin{align}\notag
		&\EE\big[ \mathcal B(\mathbf{x}_{k+1}) \,\,\,\big|\,\,\, \mathbf{x}_k,u_k\big]
		\\\notag &~~~~=\EE\big[g(x_{k+1})+ \sum_{i=1}^h \tilde{g}(x_{k-i+1})  \,\,\,\big|\,\,\, \mathbf{x}_k,u_k\big] 
		\\\label{telescopic5}&~~~~=\EE\big[g(x_{k+1}) \,\,\,\big|\,\,\, \mathbf{x}_k,u_k] \!+\! \EE\big[\sum_{i=1}^h \tilde{g}(x_{k-i+1}) \,\,\,\big|\,\,\, \mathbf{x}_k,u_k\big].
	\end{align}
	According to~\eqref{history}, and given the conditional expectation with respect to $\mathbf{x}_k$ and $u_k$, the expectation operator can be omitted since each state is conditioned on itself, \emph{i.e.,}
	\begin{equation}\label{telescopic7}
		\EE \big[\sum_{i=1}^h \tilde{g}(x_{k-i+1}) \,\,\,\big|\,\,\, \mathbf{x}_k,u_k \big] = \sum_{i=1}^h \tilde{g}(x_{k-i+1}).
	\end{equation}
	From \eqref{telescopic7}, we obtain the equivalent form of \eqref{telescopic5} as
	\begin{align}\notag
		&\EE\big[ \mathcal B(\mathbf{x}_{k+1}) \,\,\,\big|\,\,\, \mathbf{x}_k,u_k\big]\\\notag
		&~~~~=\EE\big[g(x_{k+1})  \,\,\,\big|\,\,\, \mathbf{x}_k,u_k\big] 
		+\sum_{i=1}^h \tilde{g}(x_{k-i+1}) \\\label{telescopic8}
		&~~~~=\EE\big[g(x_{k+1}) \,\,\,\big|\,\,\, \mathbf{x}_k,u_k \big] 
		+\tilde{g}(x_k)
		+\sum_{i=2}^{h} \tilde{g}(x_{k-i+1}).
	\end{align}
	Given $\mathbf{x}_k$ in \eqref{history}, by subtracting $\mathcal{B}(\mathbf{x}_k)$ from both sides of~\eqref{telescopic8} and re-indexing the summation over delayed states, we have
	\begin{align*}
		&\EE\big[ \mathcal B(\mathbf{x}_{k+1})  \,\,\,\big|\,\,\, \mathbf{x}_k,u_k\big] - \mathcal B(\mathbf{x}_k)  \\ 
		&~~~~= \EE\big[g(x_{k+1})  \,\,\,\big|\,\,\, \mathbf{x}_k,u_k\big] 	+ \tilde{g}(x_k)  + \underbrace{\sum_{i=1}^{h-1} \tilde{g}(x_{k-i})}_{\text{re-indexed}}\\[-5ex] &~~~~~  -\overbrace{\big(g(x_k) + \sum_{i=1}^h \tilde{g}(x_{k-i})\big)}^{\mathcal B(\mathbf{x}_k)}.
	\end{align*}
	By cancellation of the overlapping terms in the sums over \( \tilde{g}(x_{k-i}) \), we simplify the series and obtain
	\begin{align}\notag
		&\EE\big[ \mathcal B(\mathbf{x}_{k+1})  \,\,\,\big|\,\,\, \mathbf{x}_k,u_k\big] - \mathcal B(\mathbf{x}_k)  \\\notag  
		&~~~= \EE\big[g(x_{k+1})  \,\,\,\big|\,\,\, \mathbf{x}_k,u_k\big] 
		+ \tilde{g}(x_k) - g(x_k) - \tilde{g}(x_{k-h}).
	\end{align}
	From \eqref{history-simple}, we set \( x_{k-h} = x_h \) as the state delayed by \( h \) steps and \( x_k = x \) as the current state. Then, by substituting ${x}_{k+1}$ with the dynamics of dt-SNPS-td in~\eqref{eq:dt-SNPS-td}, we obtain      
	\begin{align}\notag
		&\EE\big[ \mathcal B(\mathbf{x}_{k+1})  \,\,\,\big|\,\,\, \mathbf{x}_k=\mathbf{x},u_k=u\big] - \mathcal B(\mathbf{x}_k)  \\\notag &~~~~= \EE\big[g(\mathcal Ax  + \mathcal A_1 x_h  +   \mathcal G u+ Ew)  \,\,\,\big|\,\,\, \mathbf{x},u\big] 
		+ \tilde{g}(x) \\\label{telescopic9} & \qquad - g(x)  - \tilde{g}(x_h).
	\end{align}
	{Since condition~\eqref{eq:sos5}  is an SOS polynomial, one has
		\begin{align*}
			&-\!\EE\!\left[
			g(\mathcal Ax  + \mathcal A_1 x_h  +   \mathcal G u+ Ew)
			\,\,\,\big|\,\,\,\mathbf{x},u\right]
			\!+\! g(x) \!-\! \tilde{g}(x)\\ &~~ \!+\! \tilde{g}\!\left(x_h\right) \!+\! \eta
			\!-\! \sum_{q=1}^{m}\big(u_q  \!-\! \mathcal{Y}_{u_q}(x,x_h) \big)
			\!-\! \mathcal{Y}^{\top}(x,x_h,u)\,\mathcal{J}(x)\notag\\
			&~~ - \mathcal{Y}^{\top}_h(x,x_h,u)\,\mathcal{J}(x_h) -\sum_{j=1}^{J}(1-b^\top_ju)\tilde{\mathcal{Y}}_j(x,x_h,u) \geq 0.
	\end{align*}}
As the terms $\mathcal{Y}^{\top}(x,x_h,u)\mathcal{J}(x)$, $\mathcal{Y}^{\top}_h(x,x_h,u)\mathcal{J}(x_h)$, $\sum_{j=1}^{J}(1-b^\top_ju) \allowbreak \tilde{\mathcal{Y}}_j(x,x_h,u)$ are all non-negative, and by selecting inputs $u_q=\mathcal{Y}_{u_q}(x,x_h)$, $q=1,\dots,m$, which results in a zero outcome for the term $\sum_{q=1}^{m}\big(u_q  - \mathcal{Y}_{u_q}(x,\allowbreak x_h) \big)$ in \eqref{eq:sos5}, one can verify that satisfaction of condition \eqref{eq:sos5} with the non-negative terms $\mathcal{Y}^{\top}(x,x_h,u)\mathcal{J}(x)$, $\mathcal{Y}^{\top}_h(x,x_h,u)\allowbreak\mathcal{J}(x_h)$, and $\sum_{j=1}^{J}(1-b^\top_ju) \allowbreak \tilde{\mathcal{Y}}_j(x,x_h,u)$ implies condition~\eqref{telescopic9}, and accordingly\allowbreak~\eqref{subeq:decreasing}. Since condition~\eqref{eq:sos4} is an SOS polynomial, the satisfaction of condition~\eqref{eq:sos4} with the non-negative terms $\mathcal{Y}_j(x,x_h)\mathcal{J}(x)$ and $\mathcal{Y}_{h_j}(x,x_h)\mathcal{J}(x_h)$, $j=1,\dots,J$ implies that the input constraints in~\eqref{input_set} are also met. It establishes $\mathcal{B}(\mathbf{x})= g(x) + \sum_{i=1}^{h} \tilde{g}(x_i)$ as a K-PCBC, which completes the proof.$\hfill\blacksquare$

\begin{remark} 
	It is worth noting that a bilinearity initially exists between the decision variables of the K-PCBC and those of the controller when evaluating the expectation in~\eqref{subeq:decreasing} over the underlying dynamics. We address this bilinearity by introducing the polynomial $\mathcal{Y}_{u_q}(x,x_h)$, which can be a state-feedback controller of the form given in \eqref{feedback}, specifically $\mathcal{Y}_u(x,x_h) = F(x_k, x_{k-h}) x_k + F_1(x_k, x_{k-h}) x_{k-h}$, as the controller in~\eqref{eq:sos5}. This reformulation allows $u$ in $\mathcal Ax  + \mathcal A_1 x_h  +   \mathcal G u+ Ew$ to serve as a symbolic variable analogous to $x$, rather than as a polynomial controller with undetermined coefficients, while the explicit polynomial $\mathcal{Y}_{u_q}(x,x_h)$ acts as the controller to be designed. {While this approach introduces conservativeness into the controller synthesis problem due to the enforcement of the condition in~\eqref{eq:sos5} with $u_q=\mathcal{Y}_{u_q}(x,x_h)$, $q=1,\dots,m$,} it offers the advantage of resolving the bilinear complexity.
\end{remark}

\begin{remark}
	There is no restriction on the choice of noise distribution in the K-PCBC formulation and it can be selected arbitrarily. In such cases, the expected value in~\eqref{eq:sos5} should be computed with respect to the chosen distribution when solving the problem using \textsf{SOSTOOLS} .
\end{remark}

\begin{remark}
	Our definitions of K-QCBC and K-PCBC accommodate multiple unsafe regions. In this context, condition \eqref{subeq:unsafe} should be repeated for all unsafe sets (cf. case studies).
\end{remark}

\begin{algorithm}[t]
	\caption{Computation of K-PCBC and safety controller}
	\label{Alg2}
	\begin{algorithmic}[1]
		\REQUIRE dt-SNPS-td $\Sigma$, $\Upsilon = (X_a, X_b, \mathcal T)$
		\IF{\eqref{eq:sos2}-\eqref{eq:sos4}, using \textsf{SOSTOOLS}~\citep{prajna2004sostools}
			and \textsf{SeDuMi}~\citep{Sturm}, are satisfied}
		\RETURN  $\gamma_a$, $\gamma_b$,\\ $[\mathcal{Y}_{u_1}(x,x_h); \mathcal{Y}_{u_2}(x,x_h); \ldots; \mathcal{Y}_{u_m}(x,x_h)]$
		\ELSE
		\STATE Adjust the degree of K-PCBC $\mathcal{B}(\mathbf{x})$
		\ENDIF
		\STATE Compute the {safety guarantee $1-\mu_h = 1-\frac{\gamma_a + \eta \mathcal{T}}{\gamma_b}$}
		\ENSURE K-PCBC $\mathcal{B}(\mathbf{x})= g(x) + \sum_{i=1}^{h} \tilde{g}(x_i)$, safety controller $u=[\mathcal{Y}_{u_1}(x,x_h); \mathcal{Y}_{u_2}(x,x_h); \ldots; \mathcal{Y}_{u_m}(x,x_h)]$, and probabilistic safety guarantee~in~\eqref{eq:safety_prob}
	\end{algorithmic}
\end{algorithm}

We present Algorithm~\ref{Alg2}, which details the steps for {computing the K-PCBC} and the corresponding safety controllers with input constraints.

\begin{table*}[t!]
	\makeatletter
	\long\def\@makecaption#1#2{%
		\vskip\abovecaptionskip
		\noindent \textbf{#1.} #2\par
		\vskip\belowcaptionskip}
	\makeatother
	\centering
	\caption{Overview of results for dt-SNPS-td systems, where $\mathcal{T}$ is the safety horizon, {$\mathsf {SD}$ is the system degree,} $n$ the state dimension, $h$ the delay, $\eta$ as in~\eqref{subeq:decreasing}, $\gamma_a$, and $\gamma_b$ define level sets, $\mu_h$ and $\alpha$ (cf.~\eqref{con2}–\eqref{con3}) denote the {safety violation bound} and tuning factor, and $\mathsf{RT}$ (s) and $\mathsf{MU}$ (MB) report the computational time and memory usage for solving the SOS problems via Algorithms~\ref{Alg1} and~\ref{Alg2}.}
\label{tab:system-configurations}
\begin{tabular}{@{}llccccccccccc@{}}
	\toprule
	Case study & Method & $\mathcal{T}$& $\mathsf {SD}$ &  $n$& $h$ & $\eta$ & $\gamma_a$ & $\gamma_b$ & $\alpha$ & $1-\mu_h$ & \(\mathsf{RT}\) & $\mathsf{MU}$ \\ 
	\midrule\midrule
	
	\multirow{2}{*}{Academic system} 
	& K-QCBC & $40$ & $2$ &  $2$ & $3$ &  $0.001$ &   $0.01$ &  $0.64$  &  $0.00001$ & {$0.9$} & $2.41$ & $1.6$\\
	& K-PCBC & $40$ & $2$ &  $2$ & $3$ &  $6.31 \times 10^{-4}$      &     $0.1$    &  $  11.08$     &  ---       &    {$0.98$}   & $33.97$    & $1.73$   \\
	\midrule
	
	\multirow{2}{*}{Jet engine compressor} 
	& K-QCBC & $60$  & $3$ & $2$ &   $4$  &  $7.65 \times 10^{-4}$ &  $0.01$ &  $0.83$ &  $0.00012$  & {$0.92$} & $2.46$  & $1.67$ \\
	& K-PCBC & $60$  & $3$ & $2$ &   $4$  &  $ 5.81 \times 10^{-4}$ &  $0.21$    &  $ 10.73$    &  ---        & {$0.97$}     & $34.12$     & $1.79$     \\
	\midrule
	
	\multirow{2}{*}{Spacecraft}      
	& K-QCBC & $20$  & $2$ & $3$ &   $3$  & $0.003$  & $0.25$  & $3.77$ & $0.00001$ & {$0.91$} & $378.6$  & $3.61$   \\
	& K-PCBC & $20$  & $2$ & $3$ &   $3$  & $5.67 \times 10^{-4}$      &  $0.76$    & $ 38.22$  &   ---    & {$0.97$}  & $411.62$      & $4.26$     \\
	\bottomrule
\end{tabular}
\end{table*}

\subsection{Discussion}\label{sec:Discussion}
We now provide a discussion on the comparative characteristics of the proposed K-QCBC and K-PCBC formulations. Between K-QCBC~\eqref{k-CBC} and K-PCBC~\eqref{K-PCBC_temp}, the latter is generally more computationally expressive and more likely to yield feasible safety certificates and corresponding controllers with a higher probability. This advantage arises from allowing polynomial, rather than strictly quadratic barrier functions, providing greater flexibility in capturing nonlinear safety boundaries. However, this expressiveness comes at a computational cost: K-PCBC typically leads to higher-order optimization problems with more decision variables, increasing both solution time and memory requirements (cf. Table~\ref{tab:system-configurations}). In contrast, K-QCBC represents a structured quadratic subclass whose algebraic simplicity enables the reformulation of safety conditions into tractable matrix inequalities in~\eqref{K-PCBC}. This structure enhances scalability and computational tractability, albeit at the expense of additional conservatism relative to K-PCBC.

Furthermore, when searching for K-QCBC, incorporating input constraints results in conditions~\eqref{con2}–\eqref{con3}, which can potentially introduce conservatism due to their sensitivity to the tuning parameter $\alpha$. Conversely, omitting the input constraints relaxes these conditions (cf.~\eqref{con1-R}), thereby reducing conservatism and improving the feasibility of the resulting optimization problem, albeit at the cost of not explicitly enforcing input constraints.

\begin{figure*}[t]
	\centering
	\subfloat[\label{fig:Aa}]{
		\includegraphics[width=0.23\textwidth]{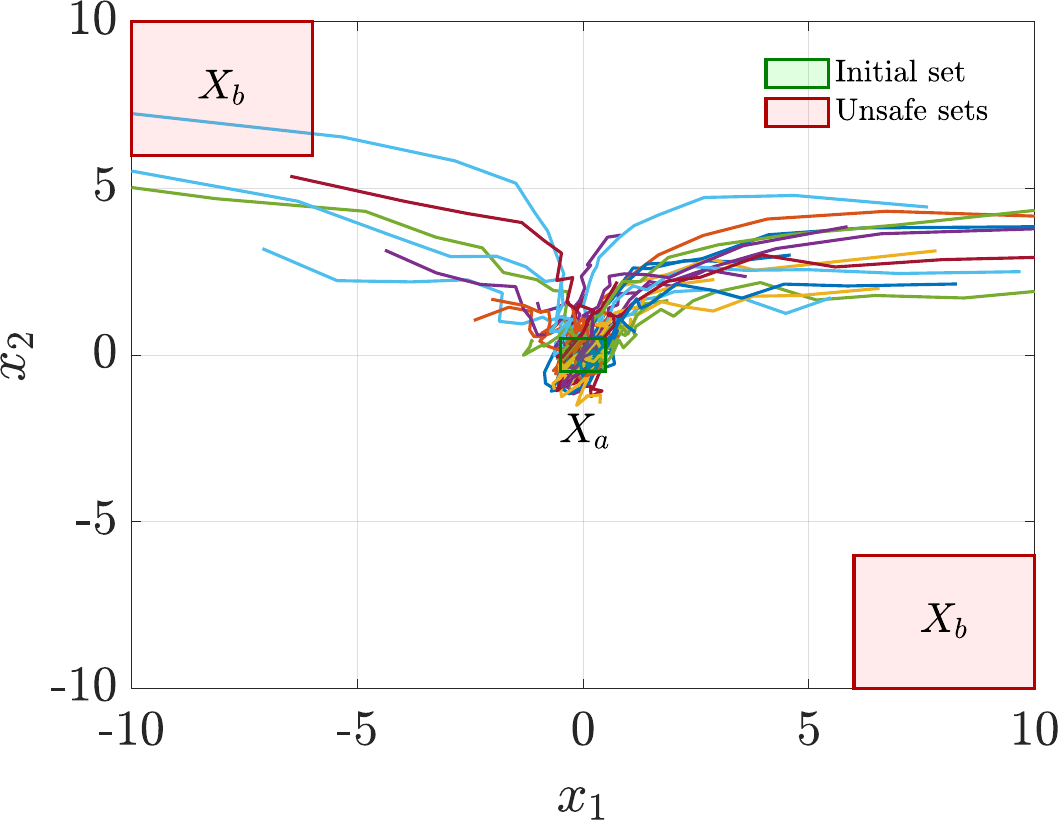}
	}\hspace{0.1cm}
	\subfloat[\label{fig:Ab}]{
		\includegraphics[width=0.23\textwidth]{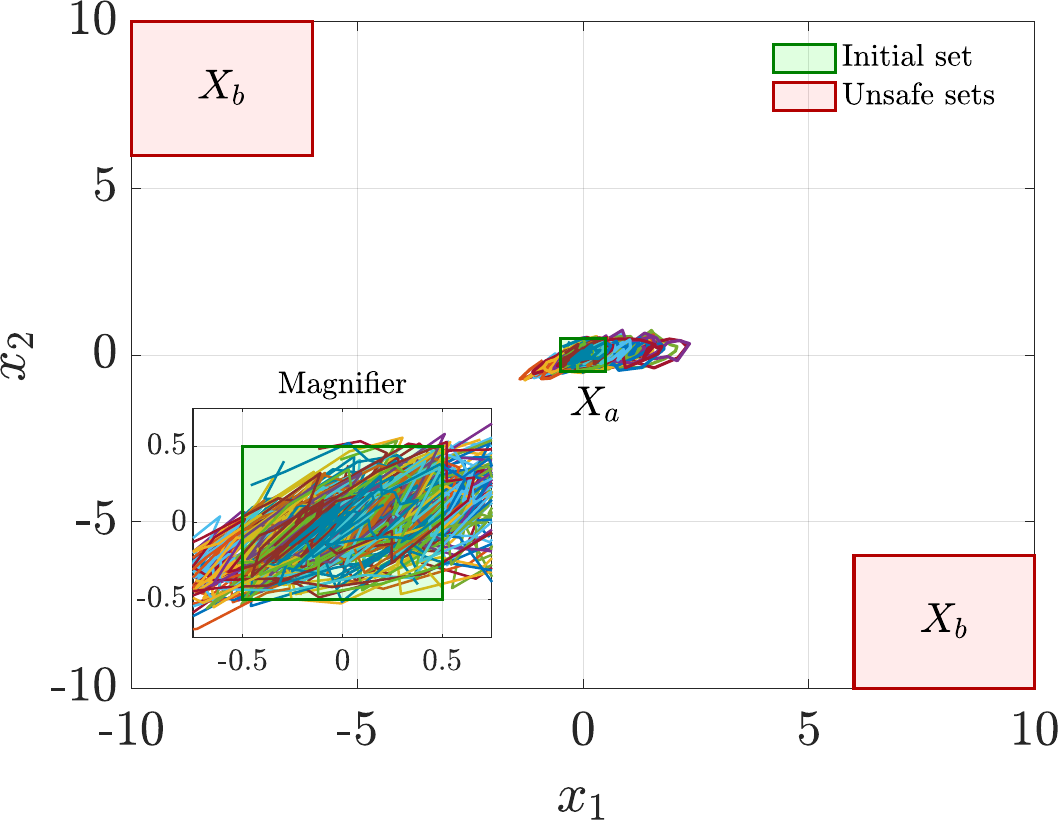}
	}\hspace{0.1cm}
	\subfloat[\label{fig:Ac}]{
		\includegraphics[width=0.23\textwidth]{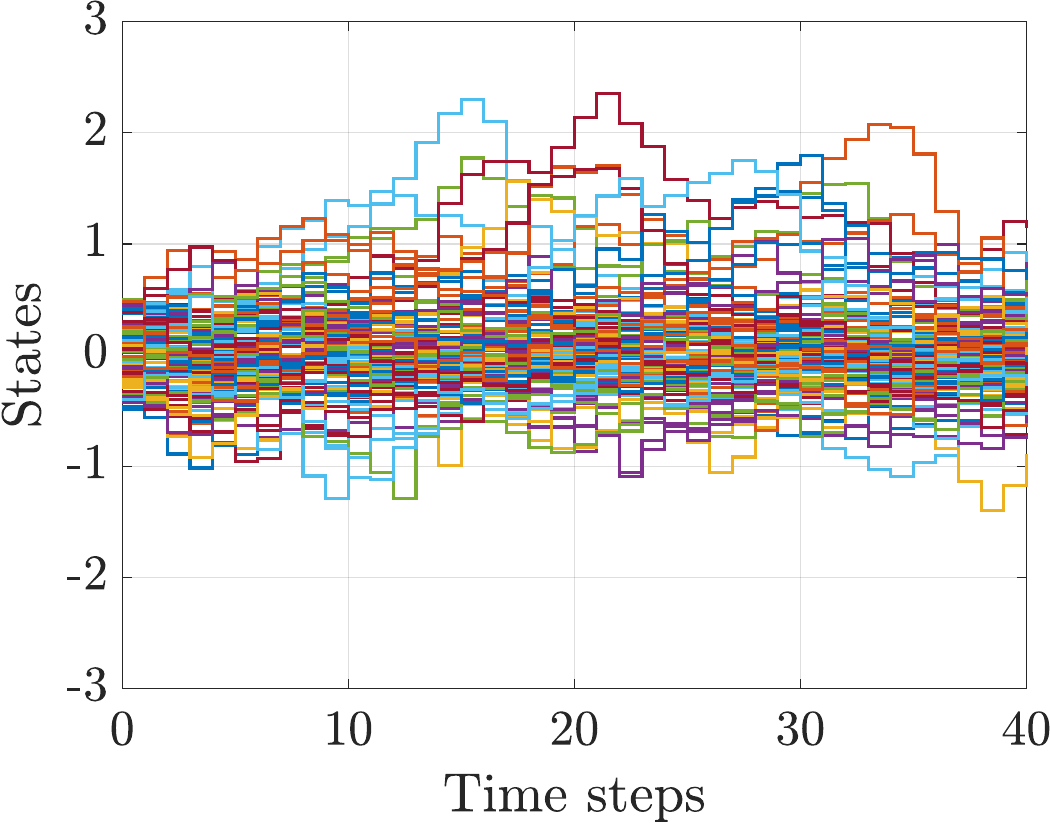}
	}\hspace{0.1cm}
	\subfloat[\label{fig:Ad}]{
		\includegraphics[width=0.232\textwidth]{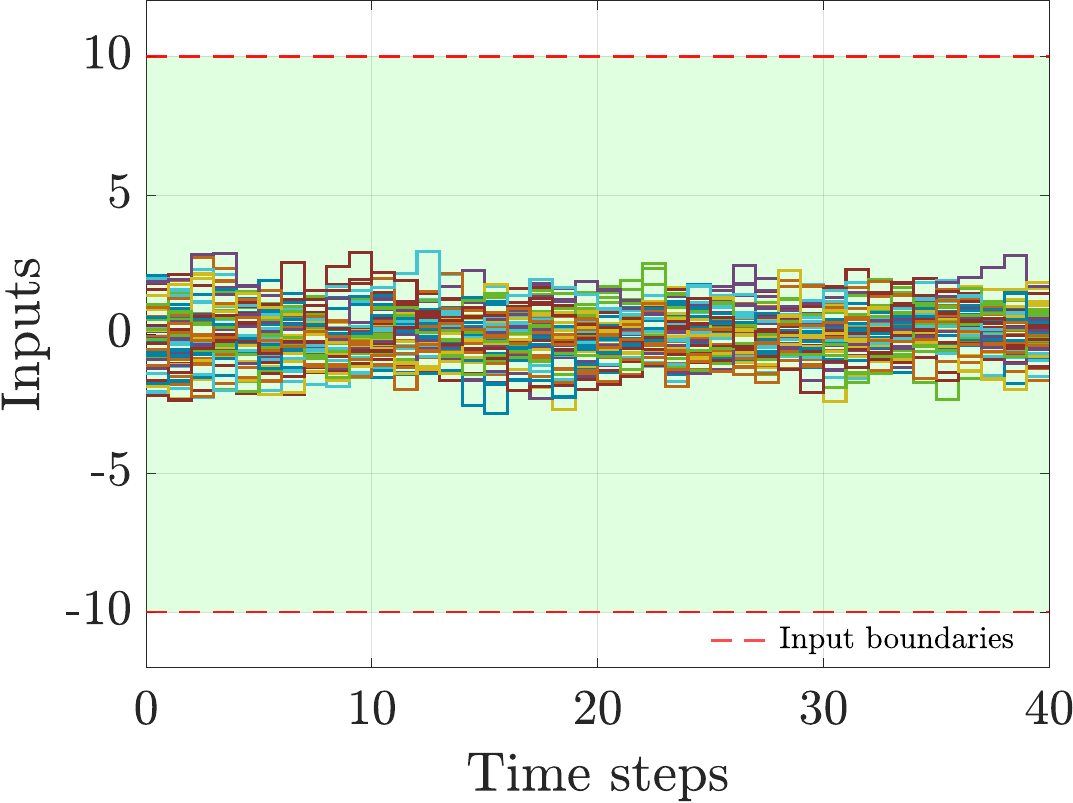}
	}
	\caption{ K-QCBC (academic system). Plot (a) illustrates the trajectories with a random input, while plot (b) displays the trajectories with the designed safety controller in~\eqref{safety-c-1}. Each simulation is generated with $50$ different noise realizations. Plot (c) depicts the trajectories over $40$ time steps, demonstrating compliance with the specified safety property $\Upsilon$. Additionally, plot (d) shows the adherence to the input constraints in \eqref{input_set}.}
	\label{fig:Atraj}
\end{figure*}

\begin{figure*}[t]
	\centering
	\subfloat[\label{fig:Ja}]{
		\includegraphics[width=0.23\textwidth]{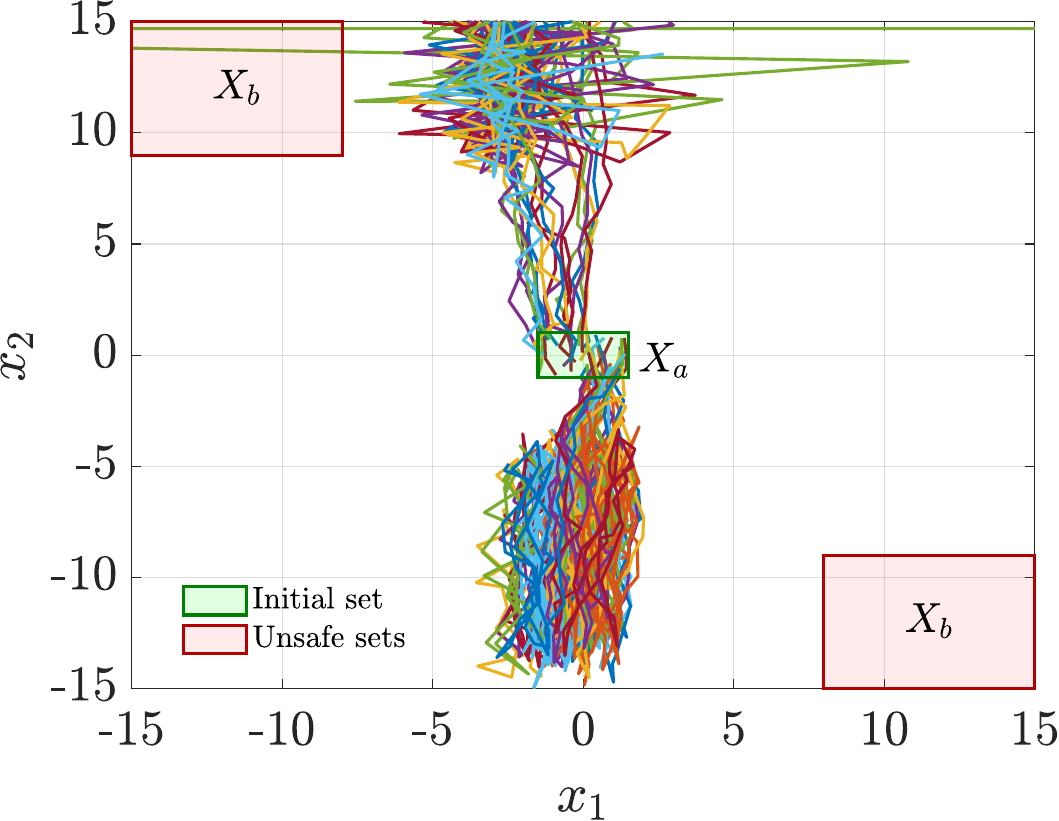}
	}\hspace{0.1cm}
	\subfloat[\label{fig:Jb}]{
		\includegraphics[width=0.23\textwidth]{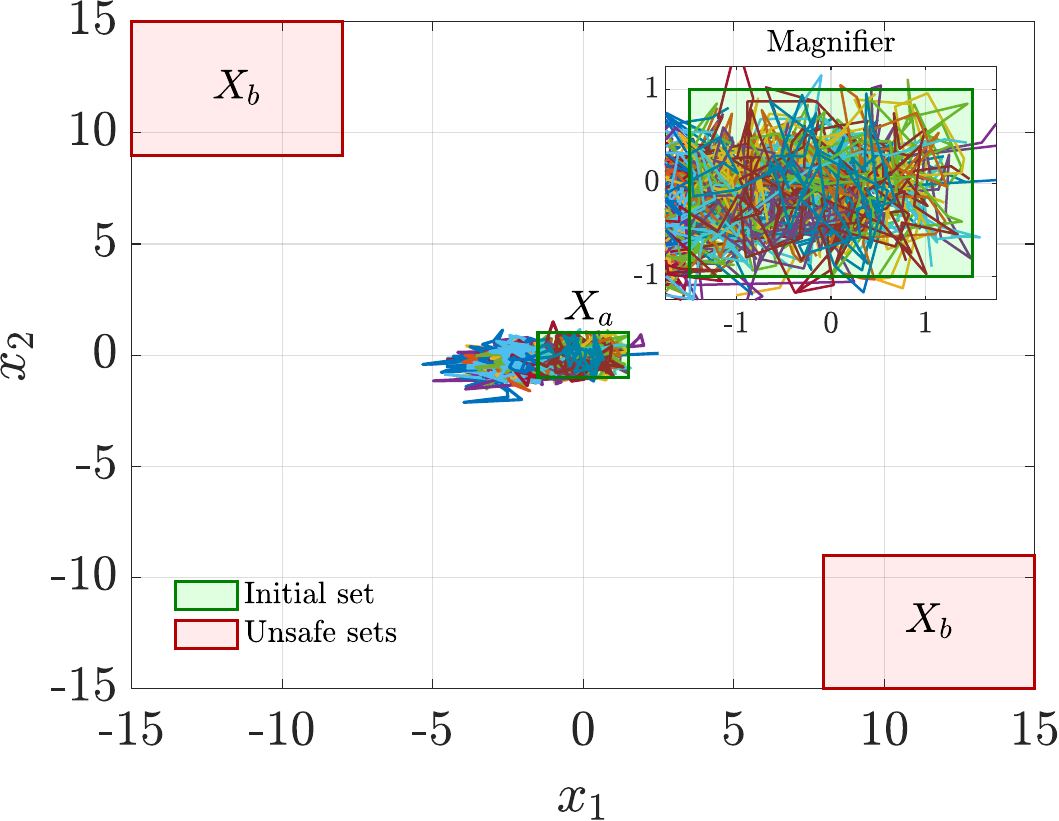}
	}\hspace{0.1cm}
	\subfloat[\label{fig:Jc}]{
		\includegraphics[width=0.23\textwidth]{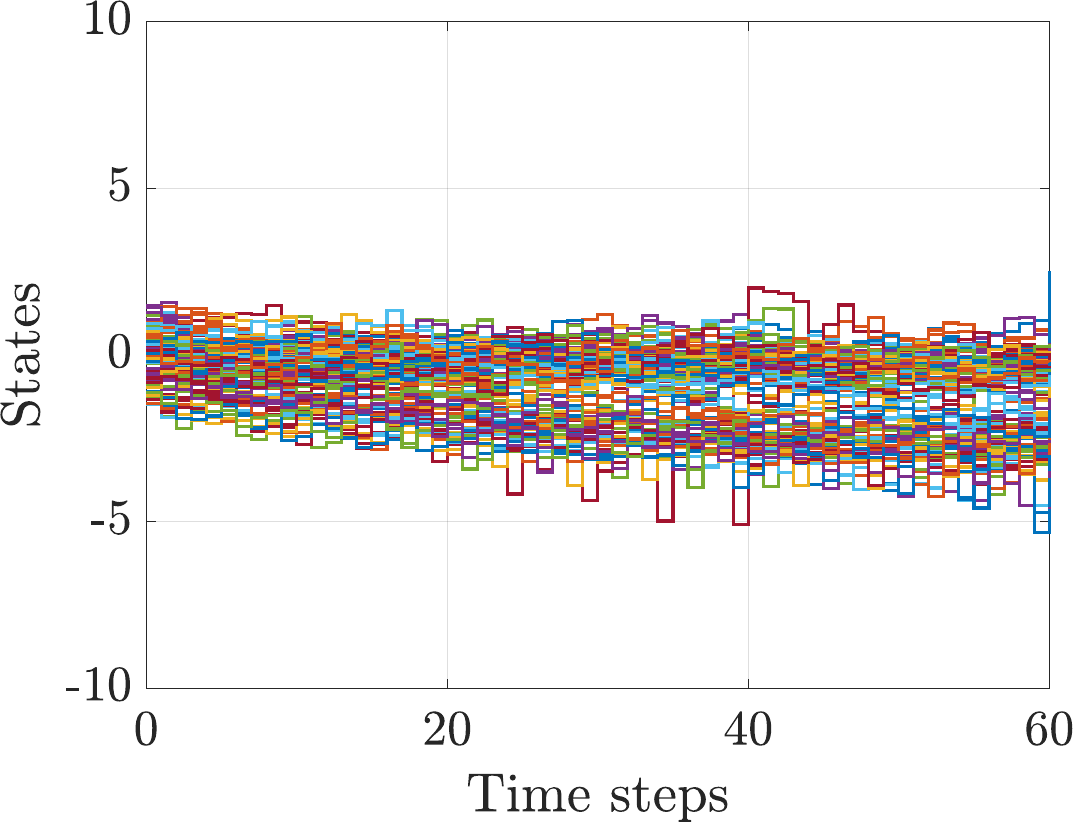}
	}\hspace{0.1cm}
	\subfloat[\label{fig:Jd}]{
		\includegraphics[width=0.23\textwidth]{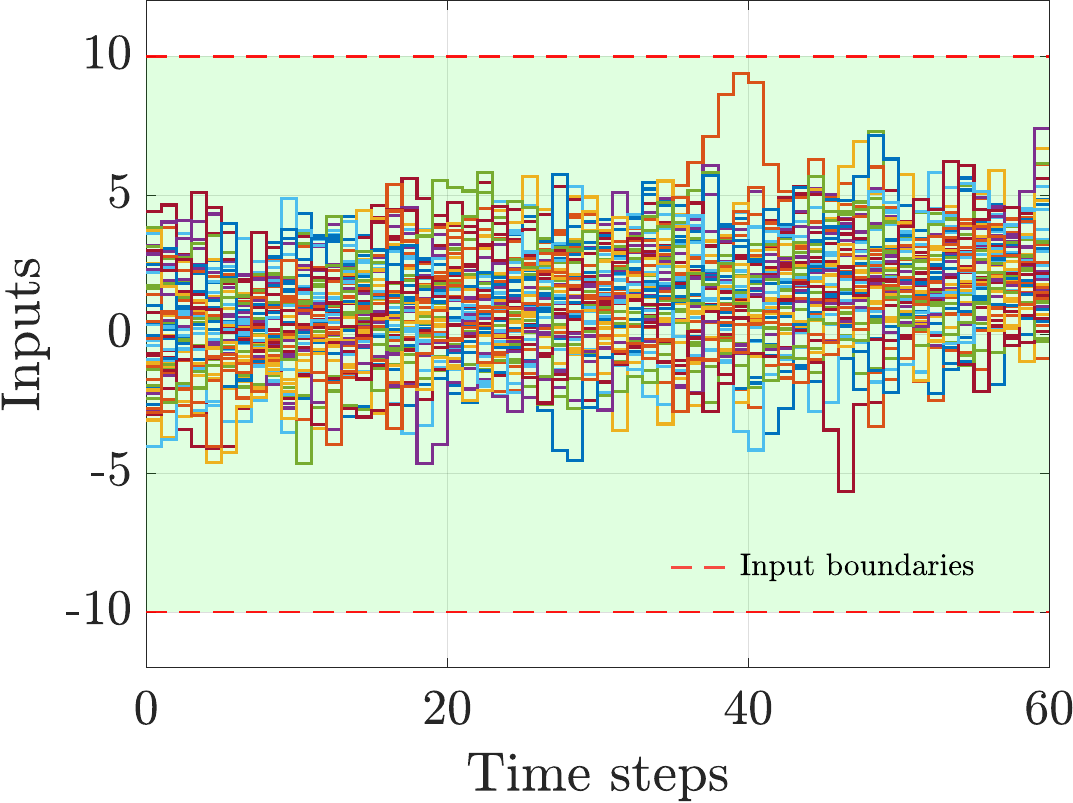}
	}
	\caption{ K-QCBC (jet engine compressor). Plot (a) shows trajectories with random input, while plot (b) includes those with the designed safety controller in~\eqref{safety-c-2}. Each simulation uses $50$ noise realizations. Plot (c) illustrates the compressor's trajectories over $60$ time steps, adhering to the safety property $\Upsilon$, while plot (d) demonstrates compliance with the input constraints in \eqref{input_set}.}
	\label{fig:Jtraj}
\end{figure*}

\begin{figure*}[t]
	\centering
	\subfloat[\label{fig:Sa}]{
		\includegraphics[width=0.23\textwidth]{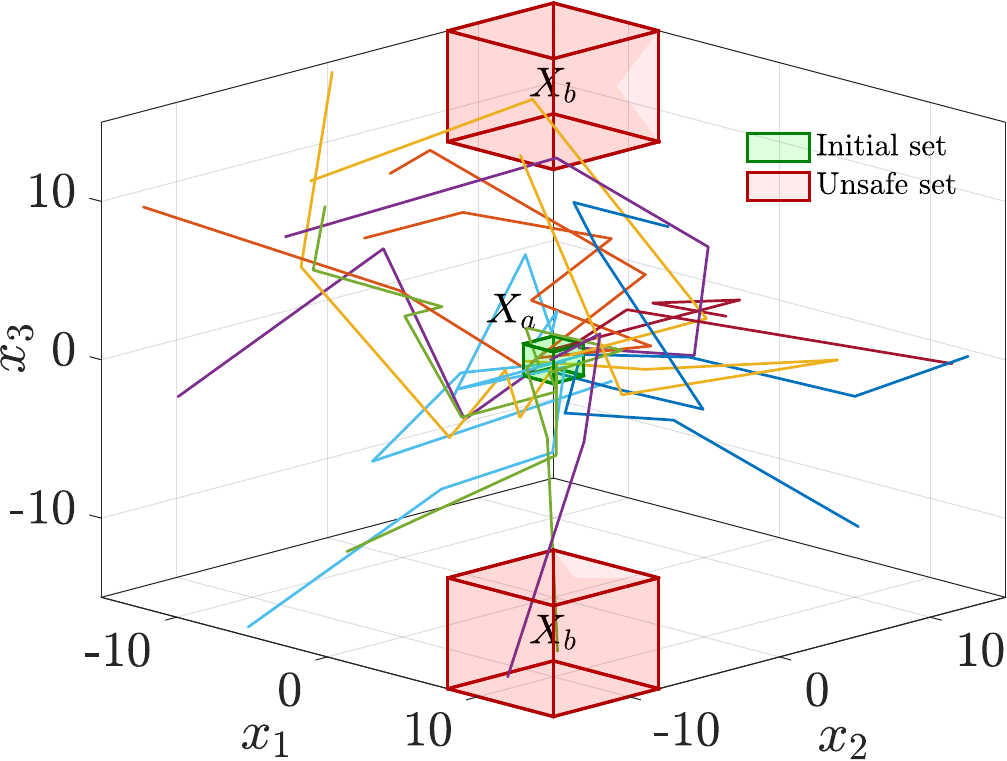}
	}\hspace{0.1cm}
	\subfloat[\label{fig:Sb}]{
		\includegraphics[width=0.23\textwidth]{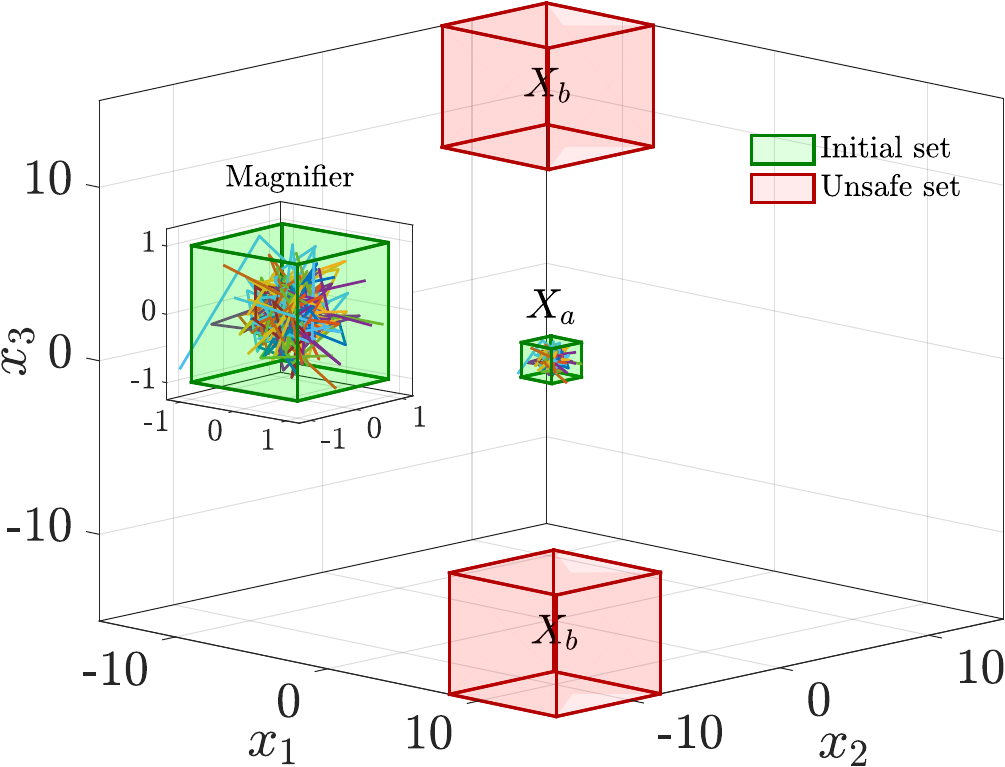}
	}\hspace{0.1cm}
	\subfloat[\label{fig:Sc}]{
		\includegraphics[width=0.23\textwidth]{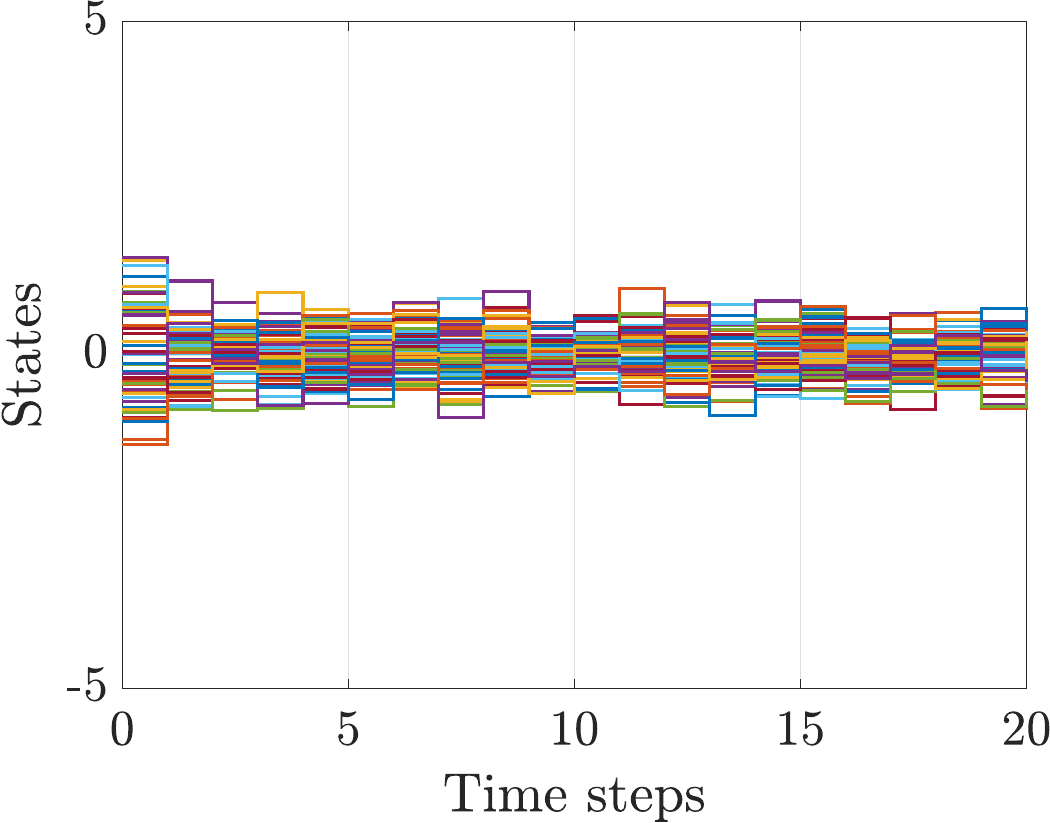}
	}\hspace{0.1cm}
	\subfloat[\label{fig:Sd}]{
		\includegraphics[width=0.23\textwidth]{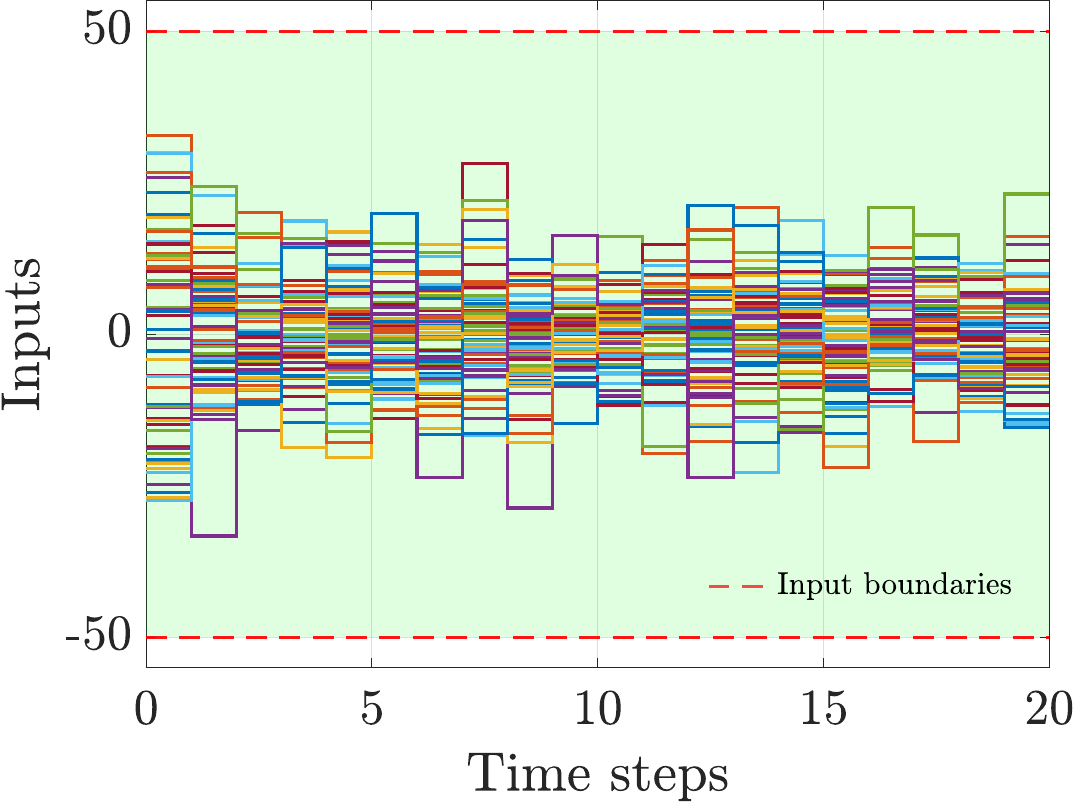}
	}
	\caption{ K-QCBC (spacecraft). Plot (a) shows the trajectories resulting from a random input, while plot (b) features the trajectories governed by the designed safety controller in~\eqref{safety-c-3}. Each simulation is executed with $20$ distinct noise realizations. Plot (c) illustrates the trajectories over $20$ time steps, confirming compliance with the defined safety property $\Upsilon$. In addition, plot (d) demonstrates the observance of the input constraints outlined in \eqref{input_set}.}  
	\label{fig:Straj}
\end{figure*}

\begin{figure*}[t]
	\centering
	\subfloat[\label{fig:Aa1}]{
		\includegraphics[width=0.23\textwidth]{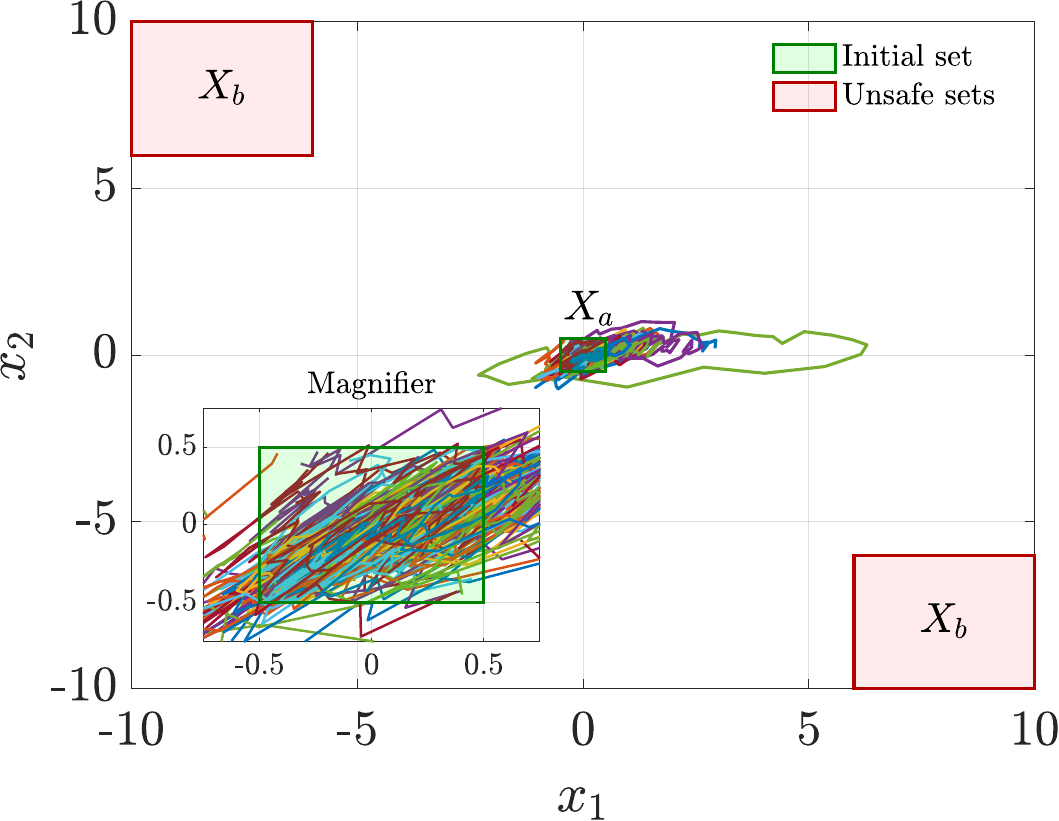}
	}\hspace{0.7cm}
	\subfloat[\label{fig:Ab1}]{
		\includegraphics[width=0.23\textwidth]{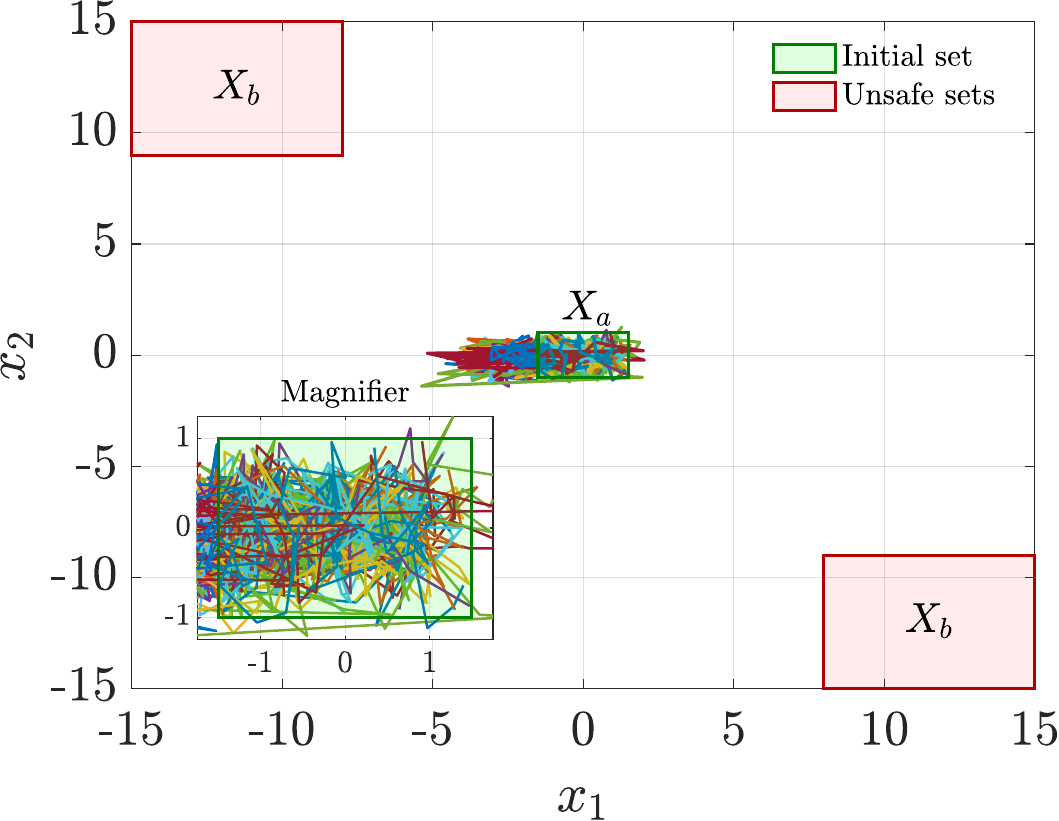}
	}\hspace{0.7cm}
	\subfloat[\label{fig:Ac1}]{
		\includegraphics[width=0.23\textwidth]{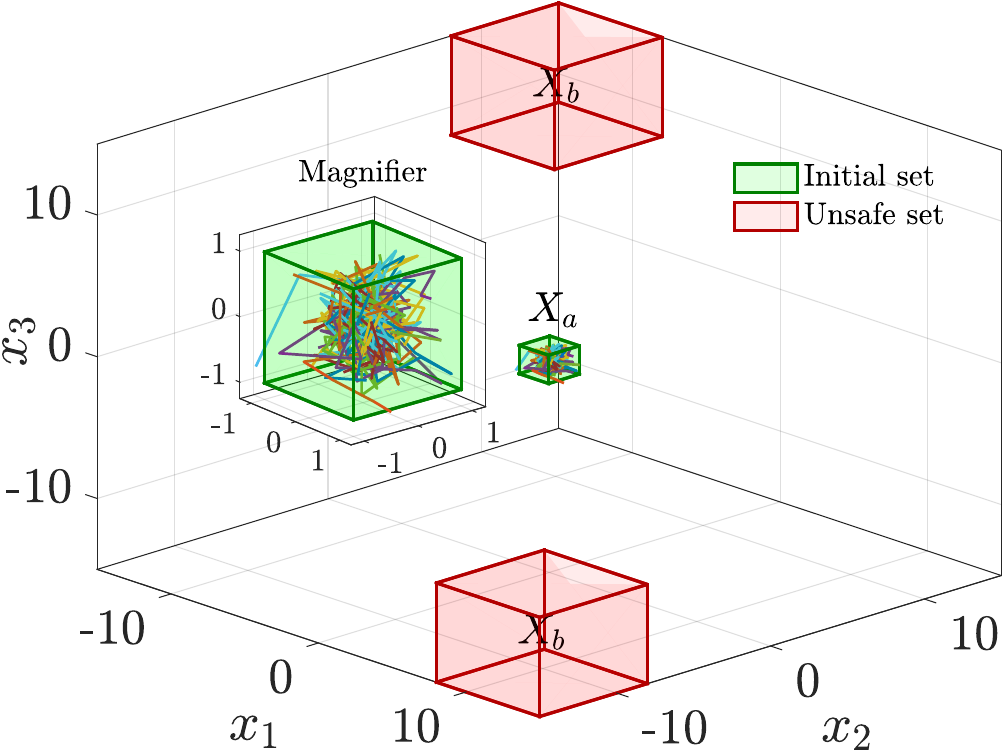}
	}
	\caption{K-PCBC (all three case studies). Plot (a) shows the academic system's trajectories with the safety controller in \eqref{safety-C-KPCBC1}. Plot (b) presents the jet engine compressor's trajectories using the safety controller in \eqref{safety-C-KPCBC2}. Plot (c) depicts the spacecraft's trajectories with the safety controller in \eqref{safety-C-KPCBC3}.}
	\label{fig:Atraj-PCBC}
\end{figure*}

\section{Case Studies}\label{sec:Case}
We demonstrate the effectiveness of our proposed approaches on a set of case studies, which include an academic system, a jet engine compressor~\citep{Tabuada-Jet}, and a spacecraft~\citep{khalil2002control}, while adapting them to accommodate delays. The primary objective across all case studies is to design K-QCBC and K-PCBC, along with their corresponding safety controllers, by employing the detailed steps of Algorithm~\ref{Alg1} and Algorithm~\ref{Alg2}. A summary of the key results is presented in Table~\ref{tab:system-configurations}. According to this table, as discussed in Discussion subsection, while K-PCBC offers a higher lower bound on the safety probability compared to K-QCBC, it also results in increased running time ($\mathsf{RT}$) and memory usage ($\mathsf{MU}$) for solving the SOS problem for each system using Algorithm~\ref{Alg2}. It is worth noting that, despite the conservatism of conditions \eqref{con2}-\eqref{con4} corresponding to the approach for handling input constraints for K-QCBC, the set of proposed conditions in \eqref{K-PCBC} remains feasible for all case studies.
All simulations were conducted on a \textsf{MacBook M2 chip} with 32~\textsf{GB} of memory. 

\subsection{Case study 1: Academic system}
We consider an academic system whose dynamics are described in Appendix~\ref{A1}. The regions of interest are given by $X= [-10,10]^2$, $X_{a} = [-0.5,0.5]^2$, $X_{b} = [6,10]\times[-10,6] \cup [-10,-6]\times[6,10],$ and $U=[-10,10].$  Through the satisfaction of conditions  \eqref{con2}, \eqref{con3}, and \eqref{L1}-\eqref{L4}, following the steps outlined in Algorithm~\ref{Alg1}, we compute K-QCBC matrices
\begin{align*}
P =	\begin{bmatrix}
	0.01 & 0\\ 0 & 0.01
\end{bmatrix}\!\!, \quad 
P_1=\begin{bmatrix}
	0.005 & 0\\ 0 & 0.005
\end{bmatrix}\!\!,
\end{align*}
and the corresponding safety controller as
\begin{align}\notag
u_k&=
0.00036 x_{k_1}^2
- 0.00041 x_{k_1} x_{k_2}
- 0.01 x_{k_1} x_{(k-3)_1}
\\\notag & ~~~- 0.003 x_{k_2}^2
-\! 0.001 x_{k_2} x_{(k-3)_1}
+ 0.01 x_{(k-3)_1}^2
\\\notag
&~~~- 0.001 x_{(k-3)_1} x_{(k-3)_2}
-\! 0.004 x_{(k-3)_2}^2
\!\!-\! 0.06 x_{k_1}\\\label{safety-c-1}
&~~~
\!\!-\! 1.57 x_{k_2}
\!+\! 0.05 x_{(k-3)_1}
\!\!-\! 0.09 x_{(k-3)_2},
\end{align}
as well as level sets $\gamma_a=0.01$, $\gamma_b=0.64$ and parameter $\eta=0.001$. According to Proposition~\ref{Th:safety}, the lower bound on the safety probability is {\(1-\mu_h=0.9 \)} over the horizon $\mathcal{T}=40$. Following Algorithm~\ref{Alg2}, we obtain the K-PCBC corresponding polynomials $g(x)$, $\tilde{g}(x_i)$, with the degree of $4$ and the safety controller as detailed in Appendix~\ref{A1}. Furthermore, the level sets $\gamma_a=0.1$, $\gamma_b=11.08$, and the parameter $\eta=6.31 \times 10^{-4}$ have been computed. According to Proposition~\ref{Th:safety}, the safety probability is guaranteed to be at least {\(1-\mu_h = 0.98\)}
over a horizon of \(\mathcal{T} = 40\) time steps. The simulation results for the academic system in \eqref{academic} using the K-QCBC and K-PCBC approaches are presented in Figures~\ref {fig:Atraj} and \ref {fig:Atraj-PCBC}, respectively.

\subsection{Case study 2: Jet engine compressor}
In the second case study, we consider a physical system, namely a jet engine compressor whose dynamics are described in Appendix~\ref{A2}. In contrast to the first case study, where the maximum degree of the system's monomials was $2$, the maximum degree for the jet engine compressor is $3$.  The regions of interest are given by $X= [-15,15]^2$, $X_{a} = [-1.5,\allowbreak1.5] \times [-1,1]$, $X_{b} = [8,15]\times[-15,-9] \cup [-15,-8] \times [9,15]$, and $U=[-10,10].$ By satisfying the conditions in \eqref{con2}, \eqref{con3}, and \eqref{L1}-\eqref{L4}, we compute the K-QCBC matrices using Algorithm~\ref{Alg1} as
\begin{align*}
P =	\begin{bmatrix}
	0.005 & 0\\ 0 & 0.005
\end{bmatrix}\!\!, 
\quad P_1=\begin{bmatrix}
	0.002 & 0\\ 0 & 0.002
\end{bmatrix}\!\!,
\end{align*}
and the corresponding safety controller as
\begin{align}\notag
u_k &= 0.007x_{k_1}^2 x_{k_2}
+ 0.008x_{k_2}^3
- 4.02x_{k_2} \\\notag&
~~~
+0.0008 x_{(k-4)_2}^2 x_{(k-4)_1} +\!0.0007,x_{(k-4)_1}^3
\!\!\\\label{safety-c-2}&
~~~- \! 0.0001x_{(k-4)_2} x_{(k-4)_1}\!\!- \! 0.39874x_{(k-4)_1},
\end{align}
along with $\gamma_a=0.01$, $\gamma_b=0.83$ and $\eta=7.65 \times 10^{-4}$. Based on Proposition~\ref{Th:safety},  {\(1- \mu_h =0.92\)} over $\mathcal{T}=60$. Following Algorithm~\ref{Alg2}, we compute the K-PCBC polynomials $g(x)$, $\tilde{g}(x_i)$, with the degree of $4$ and the associated safety controller, as detailed in Appendix~\ref{A2}, and with $\gamma_a = 0.21$, $\gamma_b = 10.73$, and $\eta = 5.81 \times 10^{-4}$. According to Proposition~\ref{Th:safety}, the safety probability is guaranteed to be at least {$1-\mu_h = 0.97$} for $\mathcal{T} = 60$. The simulation results  for the jet engine compressor using the K-QCBC and K-PCBC approaches are displayed in Figure~\ref{fig:Jtraj} and \ref {fig:Atraj-PCBC}, respectively.
\subsection{Case study 3: Spacecraft}
In the third case study, we examine a spacecraft whose dynamics are described in Appendix~\ref{A3}. Unlike the first two case studies, which feature two-dimensional dynamics, the spacecraft’s dynamics are three-dimensional. The regions of interest are given by $X= [-15,15]^3$, $X_{a} = [-1,\allowbreak1]^3$, $X_{b} = [8,15]\times[-15,-8]^2 \cup [-15,-8]\times[8,15]^2$, and $U=[-50,50]^3.$ By fulfilling the conditions specified in \eqref{con2}, \eqref{con3}, and \eqref{L1}-\eqref{L4}, we compute the K-QCBC matrices using Algorithm~\ref{Alg1} as follows
\begin{align*}
P &=	\begin{bmatrix}
	0.01 & 0 & 0\\ 0 & 0.02 & 0\\0 & 0 & 0.01
\end{bmatrix}\!\!,\,
\quad P_1 =\begin{bmatrix}
	0.009 & 0 & 0\\ 0 & 0.02 &0\\ 0 & 0 & 0.009
\end{bmatrix}\!\!,
\end{align*}
and the corresponding safety controller as
\begin{align}\notag
u_{k_1}&=
9.33x_{k_2}x_{k_3}
- 19.81x_{k_1},\\\notag
u_{k_2} &= 0.0001x_{k_1}x_{k_3}
- 19.98x_{k_2},\\\label{safety-c-3}
u_{k_3}&= - 0.07x_{k_1}x_{k_2}  - 29.93x_{k_3},		
\end{align}
as well as  \( \gamma_a =0.25\), \( \gamma_b=3.77 \), and  $\eta=0.003$. According to Proposition~\ref{Th:safety}, {\(1- \mu_h=0.91 \)} over \( \mathcal{T}=20 \). Using Algorithm~\ref{Alg2}, we compute the K-PCBC, with the degree of $4$ and the safety controller as detailed in Appendix~\ref{A3}, where $\gamma_a = 0.76$, $\gamma_b = 38.22$, and $\eta = 5.67 \times 10^{-4}$. Under Proposition~\ref{Th:safety}, {$1-\mu_h = 0.97$} over $\mathcal{T} = 20$. The simulation results for spacecraft using the K-QCBC and K-PCBC are presented in Figure~\ref{fig:Straj} and \ref {fig:Atraj-PCBC}, respectively.

\section{Conclusion}\label{sec:Conclusion}
This paper introduced a \emph{Krasovskii}-based control barrier certificate framework for discrete-time stochastic nonlinear polynomial systems with time-invariant delays. By extending the conventional delay-free CBC formulation to explicitly account for delayed dynamics, we developed two classes of barrier certificates: \emph{Krasovskii quadratic} and \emph{Krasovskii polynomial} CBC. These formulations enable the synthesis of safety controllers that provide probabilistic safety guarantees robust to system delays. The proposed safety conditions were reformulated as SOS optimization problems. Future work will focus on extending this framework to synthesize safety certificates and controllers for time-delayed systems with dynamics beyond polynomial classes.

\bibliographystyle{agsm}
\bibliography{biblio}

\section{Appendix}
\subsection{Academic system}\label{A1}
We consider an academic system with the dynamics described as
\begin{align}\notag
{x}_{(k+1)_1}& \!=\! x_{k_1} \!+\! 0.1x_{k_1} x_{k_2}+0.2 x_{(k-3)_1} x_{k_2} \! - \! 0.1x_{(k-3)_1} \\\notag &\,\,\, \,\,+0.12w_{k_1} + 0.14w_{k_2},\\\notag
{x}_{(k+1)_2} & \!=\! x_{k_2} \!-\! 0.05x_{k_1}+0.1x_{(k-3)_2} \!+\! 0.1u_{k} \!+\! 0.11w_{k_1} \\\label{academic}&~~~+0.15w_{k_2},
\end{align}
{with delay $h=3$}. The system in \eqref{academic} can be expressed in the form of dt-SNPS-td in Definition~\ref{def:dt-SNPS-td} along with the relevant matrices as 
\begin{align*}
\mathcal{A}&= \begin{bmatrix}
	1 + 0.1x_{k_2} & 0.1x_{(k-3)_1} \\
	-0.05 & 1
\end{bmatrix}\!\!,
 \mathcal{A}_1=	\begin{bmatrix}
	-0.1 + 0.1x_{k_2} & 0 \\
	0 & 0.1
\end{bmatrix}\!\!,\\ 
\mathcal{G}&=	\begin{bmatrix}
	0 \\
	0.1
\end{bmatrix}\!\!,
\quad E= \begin{bmatrix}
	0.12 && 0.14 \\
	0.11 && 0.15
\end{bmatrix}\!\!.
\end{align*}

In the following, we present the designed K-PCBC polynomials and its associated safety controller for this academic system as
\begin{align}\notag
	g(x) &= 0.131 x_{1}^{4}
	- 0.014 x_{1}^{3} x_{2}
	+ 0.005 x_{1}^{3}\notag\\
	&\quad + 0.121 x_{1}^{2} x_{2}^{2}
	- 0.004 x_{1}^{2} x_{2}
	- 0.054 x_{1}^{2}\notag\\
	&\quad - 0.014 x_{1} x_{2}^{3}
	+ 0.004 x_{1} x_{2}^{2}
	+ 0.006 x_{1} x_{2}\notag\\
	&\quad - 0.002 x_{1}
	+ 0.131 x_{2}^{4}
	- 0.005 x_{2}^{3}\notag\\
	&\quad - 0.054 x_{2}^{2}
	+ 0.002 x_{2}
	+ 0.013,\notag\\[1mm]
	\tilde{g}(x_i) &= 0.174 x_{i_1}^{4}
	- 0.011 x_{i_1}^{3} x_{i_2}
	+ 0.005 x_{i_1}^{3}\notag\\
	&\quad + 0.127 x_{i_1}^{2} x_{i_2}^{2}
	- 0.003 x_{i_1}^{2} x_{i_2}
	- 0.051 x_{i_1}^{2}\notag\\
	&\quad - 0.011 x_{i_1} x_{i_2}^{3}
	+ 0.003 x_{i_1} x_{i_2}^{2}
	+ 0.005 x_{i_1} x_{i_2}\notag\\
	&\quad - 0.002 x_{i_1}
	+ 0.174 x_{i_2}^{4}
	- 0.005 x_{i_2}^{3}\notag\\
	&\quad - 0.051 x_{i_2}^{2}
	+ 0.002 x_{i_2}
	+ 0.020,\notag\\[1mm]
	u_k &= -0.0003x_{k_1}^{2}
	+ 0.0001x_{k_2}^{2}
	+ 0.41 x_{k_1}\notag\\
	&\quad - 2.54 x_{k_2}
	- 0.03 x_{(k-3)_1}^{2}
	- 0.003 x_{(k-3)_2}^{2}\notag\\
	&\quad + 0.003 x_{k_2} x_{(k-3)_1}^{2}.\label{safety-C-KPCBC1}
\end{align}

\subsection{Jet engine compressor}\label{A2}
The jet engine compressor is defined by
\begin{align}\notag
{x}_{(k+1)_1}& \!=\! x_{k_1} \!-\! 0.1x_{(k-4)_2} \!-\! 0.15 x^2_{(k-4)_1} \!\!-\! 0.05 x^3_{k_1} \!\!+\! {0.3}w_{k_1},\\\label{jet}
{x}_{(k+1)_2} & \!=\! x_{k_2} \!+\! 0.1x_{(k-4)_1}  \!+\! 0.1u_{k} \!+\! {0.33}w_{k_2},
\end{align}
{where ${x}_{k_1}$ is the mass flow, ${x}_{k_2}$ is the pressure rise, $u_k$ corresponds to the throttle mass flow, the control input, and
delay $h=4$}. The system described in \eqref{jet} can be represented in the dt-SNPS-td format outlined in Definition~\ref{def:dt-SNPS-td} as
\begin{align*}
\mathcal{A}&= \begin{bmatrix}
	1 - 0.05x^2_{k_1} & 0 \\
	0& 1
\end{bmatrix}\!\!,
\quad\mathcal{A}_1=	\begin{bmatrix}
	-0.15x_{(k-4)_1} & -0.1 \\
	0.1 & 0
\end{bmatrix}\!\!,\\ 
\mathcal{G}&=	\begin{bmatrix}
	0 \\
	0.1
\end{bmatrix}\!\!,
\quad E= \begin{bmatrix}
	0.3 && 0 \\
	0 && 0.33
\end{bmatrix}\!\!.
\end{align*}

Here, we present the designed K-PCBC polynomials and its corresponding safety controller as follows:
\begin{align}\notag
	g(x) &= 0.135 x_{1}^{4}
	- 0.010 x_{1}^{3} x_{2}
	+ 0.007 x_{1}^{3}\notag\\
	&\quad + 0.127 x_{1}^{2} x_{2}^{2}
	- 0.004 x_{1}^{2} x_{2}
	- 0.049 x_{1}^{2}\notag\\
	&\quad - 0.010 x_{1} x_{2}^{3}
	+ 0.004 x_{1} x_{2}^{2}
	+ 0.004 x_{1} x_{2}\notag\\
	&\quad - 0.002 x_{1}
	+ 0.135 x_{2}^{4}
	- 0.007 x_{2}^{3}\notag\\
	&\quad - 0.049 x_{2}^{2}
	+ 0.002 x_{2}
	+ 0.016,\notag\\[1mm]
	\tilde{g}(x_i) &= 0.178 x_{i_1}^{4}
	- 0.008 x_{i_1}^{3} x_{i_2}
	+ 0.007 x_{i_1}^{3}\notag\\
	&\quad + 0.148 x_{i_1}^{2} x_{i_2}^{2}
	- 0.003 x_{i_1}^{2} x_{i_2}
	- 0.039 x_{i_1}^{2}\notag\\
	&\quad - 0.008 x_{i_1} x_{i_2}^{3}
	+ 0.004 x_{i_1} x_{i_2}^{2}
	+ 0.004 x_{i_1} x_{i_2}\notag\\
	&\quad - 0.002 x_{i_1}
	+ 0.177 x_{i_2}^{4}
	- 0.006 x_{i_2}^{3}\notag\\
	&\quad - 0.039 x_{i_2}^{2}
	+ 0.002 x_{i_2}
	+ 0.025,\notag\\[1mm]
	u_k &= 0.01 x_{1}^{3}
	+ 0.005 x_{k_2}^{3}
	- 5.17 x_{k_2}\notag\\
	&\quad + 0.006 x_{(k-4)_1} x_{(k-4)_2}^{2}
	+ 0.005 x_{(k-4)_1}^{3}\notag\\
	&\quad - 0.005 x_{(k-4)_1} x_{(k-4)_2}
	- 0.69 x_{(k-4)_1}.\label{safety-C-KPCBC2}
\end{align}

\subsection{Spacecraft}\label{A3}
The spacecraft is characterized by
\begin{align}
	\begin{split}\label{space} 
		{x}_{(k+1)_1} &\!=\! 	{x}_{k_1} + \frac{J_{2} - J_{3}}{J_{1}}\, x_{k_2}\, x_{k_3} 
		+ \frac{1}{J_1}\, u_{k_1}  + {0.3}w_{k_1}, \\
		{x}_{(k+1)_2} &\!=\!  {x}_{k_2} \!+\! \frac{J_{3} - J_{1}}{J_{2}}\, x_{(k-3)_1}\, x_{(k-3)_3} 
		\!+\! \frac{1}{J_{2}}\, u_{k_2}  \!+\! {0.3}w_{k_2}, \\
		{x}_{(k+1)_3} &\!=\!  {x}_{k_3} + \frac{J_{1} - J_{2}}{J_{3}}\, x_{k_1}\, x_{(k-3)_2} 
		+ \frac{1}{J_{3}}\, u_{k_3}  + {0.3}w_{k_2},
	\end{split}
\end{align}
with delay $h=3$. Additionally, $u_k=\left[u_{k_1} ; u_{k_2} ; u_{k_3}\right]$ is the torque input, and $J_{1}$ to $J_{3}$ are the principal moments of inertia. The system in \eqref{space} can be expressed in the form of dt-SNPS-td as defined in Definition~\ref{def:dt-SNPS-td} as
\begin{align*}
\mathcal{A} &= \begin{bmatrix}  1& 0 &\frac{J_{2} - J_{3}}{J_{1}}\, x_{k_2} \\0 & 1& 0  \\0& 0 & 1\end{bmatrix}\!\!, \mathcal{A}_1 = \begin{bmatrix}  0& 0 & 0  \\\frac{J_{3} - J_{1}}{J_{2}} x_{(k-3)_3}  & 0& 0  \\0& \frac{J_{1} - J_{2}}{J_{3}}\, x_{k_1} & 0\end{bmatrix}\!\!, \\
\mathcal{G} &=  \begin{bmatrix} \frac{1}{J_{1}}  & 0  & 0\\ 0 &  \frac{1}{J_{2}} &0 \\ 0 &0 &\frac{1}{J_{3}}    \end{bmatrix}\!\!,
\quad E =  \begin{bmatrix} 0.3 && 0  && 0\\ 0 &&  0.3 &&0 \\ 0 &&0 &&0.3    \end{bmatrix}\!\!.
\end{align*}

In the following, we provide the designed K-PCBC polynomials and its associated safety controller for this spacecraft systems as
\begin{align}\notag
	g(x) &=
	0.064 x_{1}^{4}
	- 0.006 x_{1}^{3} x_{2}
	- 0.006 x_{1}^{3} x_{3}\\\notag
	&\quad + 0.003 x_{1}^{3}
	+ 0.040 x_{1}^{2} x_{2}^{2}
	+ 0.004 x_{1}^{2} x_{2} x_{3}\\\notag
	&\quad - 0.001 x_{1}^{2} x_{2}
	+ 0.039 x_{1}^{2} x_{3}^{2}
	- 0.001 x_{1}^{2} x_{3}\\\notag
	&\quad - 0.124 x_{1}^{2}
	- 0.006 x_{1} x_{2}^{3}
	- 0.004 x_{1} x_{2}^{2} x_{3}\\\notag
	&\quad + 0.001 x_{1} x_{2}^{2}
	- 0.004 x_{1} x_{2} x_{3}^{2}
	+ 0.013 x_{1} x_{2}\\\notag
	&\quad - 0.006 x_{1} x_{3}^{3}
	+ 0.001 x_{1} x_{3}^{2}
	+ 0.012 x_{1} x_{3}\\\notag
	&\quad - 0.004 x_{1}
	+ 0.064 x_{2}^{4}
	+ 0.006 x_{2}^{3} x_{3}\\\notag
	&\quad - 0.003 x_{2}^{3}
	+ 0.039 x_{2}^{2} x_{3}^{2}
	- 0.001 x_{2}^{2} x_{3}\\\notag
	&\quad - 0.124 x_{2}^{2}
	+ 0.006 x_{2} x_{3}^{3}
	- 0.001 x_{2} x_{3}^{2}\\\notag
	&\quad - 0.013 x_{2} x_{3}
	+ 0.005 x_{2}
	+ 0.063 x_{3}^{4}\\\notag
	&\quad - 0.003 x_{3}^{3}
	- 0.123 x_{3}^{2}
	+ 0.005 x_{3}
	+ 0.149,\\\notag
	\tilde{g}(x_i) &=
	0.084 x_{i_1}^{4}
	- 0.006 x_{i_1}^{3} x_{i_2}
	- 0.005 x_{i_1}^{3} x_{i_3}\\\notag
	&\quad + 0.002 x_{i_1}^{3}
	+ 0.032 x_{i_1}^{2} x_{i_2}^{2}
	+ 0.003 x_{i_1}^{2} x_{i_2} x_{i_3}\\\notag
	&\quad - 0.001 x_{i_1}^{2} x_{i_2}
	+ 0.032 x_{i_1}^{2} x_{i_3}^{2}
	- 0.001 x_{i_1}^{2} x_{i_3}\\\notag
	&\quad - 0.119 x_{i_1}^{2}
	- 0.006 x_{i_1} x_{i_2}^{3}
	- 0.003 x_{i_1} x_{i_2}^{2} x_{i_3}\\\notag
	&\quad + 0.001 x_{i_1} x_{i_2}^{2}
	- 0.003 x_{i_1} x_{i_2} x_{i_3}^{2}
	+ 0.013 x_{i_1} x_{i_2}\\\notag
	&\quad - 0.005 x_{i_1} x_{i_3}^{3}
	+ 0.001 x_{i_1} x_{i_3}^{2}
	+ 0.013 x_{i_1} x_{i_3}\\\notag
		&\quad - 0.004 x_{i_1}
		+ 0.084 x_{i_2}^{4}
		+ 0.006 x_{i_2}^{3} x_{i_3}\\\notag
		&\quad - 0.003 x_{i_2}^{3}
		+ 0.032 x_{i_2}^{2} x_{i_3}^{2}
		- 0.001 x_{i_2}^{2} x_{i_3}\\\notag
		&\quad - 0.119 x_{i_2}^{2}
		+ 0.005 x_{i_2} x_{i_3}^{3}
		- 0.001 x_{i_2} x_{i_3}^{2}\\\notag
		&\quad - 0.013 x_{i_2} x_{i_3}
		+ 0.005 x_{i_2}
		+ 0.084 x_{i_3}^{4}\\\notag
		&\quad - 0.003 x_{i_3}^{3}
		- 0.118 x_{i_3}^{2}
		+ 0.005 x_{i_3}
		+ 0.163,\\\notag
u_{k_1} &= 5.15x_{k_1}^{2} - 14.54x_{k_1},\\\notag
u_{k_2} &= 0.001x_{k_3}^{2} - 12.43x_{k_2} - 0.06x_{k_3} x_{(k-3)_1},\\\label{safety-C-KPCBC3}
u_{k_3} &= - 0.02x_{k_1}^2- 10.53x_{k_3} + 0.01x_{(k-3)_1}^{2}.
\end{align}

\newpage

\begin{authorbio}[Omid]{Omid Akbarzadeh} is currently a PhD student in the School of Computing at Newcastle University, U.K. His academic journey commenced at Shiraz University, where he obtained a Bachelor of Science in Electrical and Electronic Engineering. Following this, he pursued a master's degree in Communications and Computer Network Engineering (CCNE) at the Polytechnic University of Turin, Italy (Politecnico di Torino). His research interests include safe cyber-physical systems, communication networks, data-driven approaches, and formal control.
\end{authorbio}

\begin{authorbio}[MohammadHossein]{MohammadHossein Ashoori} received his B.Sc. and M.Sc. degrees in Electrical Engineering from Sharif University of Technology (SUT), Tehran, Iran, in 2019 and 2022, respectively. He is currently pursuing his PhD in
the School of Computing at Newcastle University, UK.  His research interests include cyber-physical systems (CPS), computer vision, and digital signal processing.
\end{authorbio}

\begin{authorbio}[Amy]{Amy Nejati} is an Assistant Professor in the School of Computing at Newcastle University in the United Kingdom. Prior to this, she was a Postdoctoral Associate at the Max Planck Institute for Software Systems in Germany from July 2023 to May 2024. She also served as a Senior Researcher in the Computer Science Department at the Ludwig Maximilian University of Munich (LMU) from November 2022 to June 2023. She received the PhD in Electrical Engineering from the Technical University of Munich (TUM) in 2023. She has received the B.Sc. and M.Sc. degrees both in Electrical Engineering. Her line of research mainly focuses on developing efficient (data-driven) techniques to design and control highly-reliable autonomous systems while providing mathematical guarantees. She was selected as a Best Repeatability Prize Finalist at ACM HSCC 2025 and as one of the CPS Rising Stars \end{authorbio}

\newpage

\begin{authorbio}[Abolfazl]{Abolfazl Lavaei} is an Assistant Professor in the School of Computing at Newcastle University, United Kingdom. Between January 2021 and July 2022, he was a Postdoctoral Associate in the Institute for Dynamic Systems and Control at ETH Zurich, Switzerland. He was also a Postdoctoral Researcher in the Department of Computer Science at LMU Munich, Germany, between November 2019 and January 2021. He received the Ph.D. degree in Electrical Engineering from the Technical University of Munich (TUM), Germany, in 2019. He obtained the M.Sc. degree in Aerospace Engineering with specialization in Flight Dynamics and Control from the University of Tehran (UT), Iran, in 2014. He is the recipient of several international awards in the acknowledgment of his work including  Best Repeatability Prize (Finalist) at the ACM HSCC 2025, IFAC ADHS 2024, and IFAC ADHS 2021, HSCC Best Demo/Poster Awards 2022 and 2020, and IFAC Young Author Award Finalist 2019. His line of research primarily focuses on the intersection of Control Theory, Formal Methods, and Statistical Learning Theory.
\end{authorbio}

\end{document}